\documentclass[11pt]{article}

\include{epic.sty}
\usepackage{amssymb,amsmath,amsthm,hyperref,srcltx,accents}
\usepackage{amsfonts,stmaryrd,mathrsfs}
\usepackage{paralist}
\usepackage{tikz}
\usepackage{subfigure}
\usepackage{times}
\usepackage{array}
\usepackage{multirow}
\usepackage{graphicx}
\usepackage{gensymb}
\usepackage{algorithm,algorithmic}

\newboolean{IncludeFragments}
\setboolean{IncludeFragments}{false}

\newlength{\originalbase}
\newcommand{\spacing}[1]{\setlength{\baselineskip}{#1\originalbase}}

\newif\ifnotesw\noteswtrue

\newtheorem{theorem}{Theorem}[section]

\newtheorem{lemma}[theorem]{Lemma}

\newtheorem{cor}[theorem]{Corollary}
\newtheorem{corollary}[theorem]{Corollary}

\theoremstyle{definition}
\newtheorem{mydef}[theorem]{Definition}

\newcommand{\remove}[1]{}

\newcommand{\R}{{\mathbb R}}

\def\cP{{\mathcal P}}

\def\cQ{{\mathcal Q}}
\def\cH{{\mathcal H}}

\def\hull{\mathrm{Hull}}
\def\proj{\mathrm{proj}}

\def\retract{\pi}
\def\metricproj{\rho}

\begin{document}
%\addtolength{\baselineskip}{\baselineskip}
%\addtolength{\parskip}{\parskip}
\spacing{1.4}
\parskip=+3pt

\def\proofend{\hfill$\Box$\medskip}
\def\Proof{\noindent{\emph{Proof. }}  }
\newcommand{\ProofOf}[1]{\noindent{\emph{ Proof of {#1}. }}}

\newenvironment{proofof}[1]{\ProofOf{#1}}{\hfill $\Box$ \medskip}

\def\cop{{p}}
\def\rob{{e}}

\def\proj{\mathrm{proj}}
\def\diam{\mathrm{diam}}

\def\leapenv{
\tikzstyle{every node}=[font=\footnotesize]
% outer boundary
\draw [thick] (1.25,1.25) -- (2, 1)  -- (2.25, .25) -- (1,0) -- (1,-.25) -- (4.5, -.5) -- (4,-1.5) -- (-1, -1.5) -- (-1.25, .2)  -- (1.25,1.25);

% obstacles
%\draw[fill=gray] (-.5, -.5) -- (-.2, -.6) -- (0, -.3) -- (-.3, -.2) -- cycle;
\draw[fill=gray] (-.5, -.9) -- (-.2, -1) -- (0, -.7) -- (-.3, -.6) -- cycle;
\draw[fill=gray] (-.5, 0) -- (-.1, 0) -- (0, -.3) -- cycle;
\draw[fill=gray] (1.1, .75) -- (1.85, .5) -- (1.6, 1) -- cycle;
\draw[fill=gray] (.78, -.8) -- (1.25, -1.2) -- (1.5, -.9) -- cycle;
\draw[fill=gray] (2.5, -.8) -- (3.25, -1.2) -- (3.1, -.8) -- (2.75, -.6) -- cycle;

% frontiers
\draw[dashed]  (-.5,.5) --(-.5, -1.5);
\draw[dashed] (.5, -1.5) -- (1,-.25) -- (1,0) -- (1.17,1.21); 
\draw[dashed]  (2.5, -.35) -- (2.5,-1.5);

\draw[fill] (-.5,.51) circle (1pt);
\draw[fill] (-.5, -1.5) circle (1pt);
\draw[fill] (2.5, -.35) circle (1pt);
\draw[fill] (2.5,-1.5) circle (1pt);

% shared
%\draw[very thick] (1,0) -- (1, -.25);
%\draw[very thick]  (.5, -1.5) -- (-1, -1.5) -- (-1.25, .2) -- (1.17,1.21);

\draw[fill] (1,0) circle (1pt);
\draw[fill] (1, -.25) circle (1pt);
\draw[fill] (.5, -1.5) circle (1pt);
\draw[fill] (1.17,1.21) circle (1pt);
\draw[fill] (-1, -1.5) circle (1pt);

%\node[below] at (-1, -1.5) {$u$};
\node[below] at (-.5, -1.5) {$u$};
\node[below] at (.5, -1.5) {$v$};
\node[below] at (2.5, -1.5) {$w$};

}

%%%%%%%%%%%%%%%%

\title{A Leapfrog Strategy for Pursuit-Evasion in a Polygonal Environment}

\author{Brendan Ames\footnote{Institute for Mathematics and its Applications, University of Minnesota, Minneapolis MN 55455. Current affiliation: Department of Mathematics, University of Alabama, Tuscaloosa, AL 35487},~
Andrew Beveridge\footnote{Department of Mathematics, Statistics and Computer Science, Macalester College, St Paul MN 55105},~
Rosalie Carlson\footnote{Department of Mathematics, Harvey Mudd College, Claremont CA 91711. Current affiliation: Department of Mathematics, University of California, Santa Barbara CA 93106.},~
 Claire Djang\footnote{Department of Mathematics, Oberlin College, Oberlin OH 44074}, \\ 
 Volkan Isler\footnote{Department of Computer Science, University of Minnesota, Minneapolis MN 55455},~~
 Stephen Ragain\footnote{Department of Mathematics, Pomona College, Claremont CA 91711. Current affiliation: Department of Management Science and Engineering, Stanford University, Stanford CA 94305},~~and 
Maxray Savage$^{\dagger}$
}

\maketitle

\begin{abstract}
We study pursuit-evasion in a polygonal environment with polygonal obstacles.
In this turn based game, an evader $\rob$ is chased by pursuers $\cop_1, \cop_2, \ldots, \cop_{\ell}$. The players have full information about the environment and the location of the other players. The pursuers are allowed to coordinate their actions. On the pursuer turn, each $\cop_i$ can move to any point at distance at most 1 from his current location. On the evader turn, he moves similarly. The pursuers win if some pursuer becomes co-located with the evader in finite time. The evader wins if he can evade capture forever.

It is known that one pursuer can capture the evader in any simply-connected polygonal environment, and that three pursuers are always sufficient in any polygonal environment (possibly with  polygonal obstacles). We contribute two new results to this field. First, 
we fully characterize when an environment with a single obstacles is one-pursuer-win or two-pursuer-win.
Second, we give sufficient (but not necessary) conditions for an environment to have a winning strategy for two pursuers.  Such environments can be swept by a \emph{leapfrog strategy} in which the two cops alternately guard/increase the currently controlled area. The running time of this algorithm is $O(n \cdot h  \cdot \diam(P))$ where $n$ is the number of vertices, $h$ is the number of obstacles and $\diam(P)$ is the diameter of $P$.

More concretely, for an environment with $n$ vertices, we describe an $O(n^2)$ algorithm that (1) determines whether the obstacles are well-separated, and if so, (2) constructs the required partition for a leapfrog strategy.  
\end{abstract}

%%%%%%%%%%%%% Intro %%%%%%%%%%%%%%%

\section{Introduction}

We study a pursuit-evasion game known as the {\em lion and man game}. In this game, which takes place in  a polygonal environment,
pursuers $\cop_1, \cop_2, \ldots , \cop_{\ell}$ try to capture an evader $\rob$. The environment   consists of the polygon $P$ with some polygonal obstacles (or holes) $\cH = \{ H_1, H_2, \ldots , H_h \}.$  The players are located in $P  \,\,  \backslash  \!  \bigcup_{H_i \in \cH} H_i.$ We use $P$ to denote the environment (with obstacles), and    $\partial P$ to denote its outer boundary. Each obstacle $H_i$ is  an open set, so that the players may occupy a point on its boundary $\partial H_i$.
 A player located at point $x \in P$ can move to any point in
$
B(x,1) = \{ y \in P \mid d(x,y) \leq 1 \},
$
where $d(x,y)=d_P(x,y)$ is the length of a shortest $(x,y)$-path in $P$. Let $\rob^t$ denote the position of the evader at the end of round $t$, and similarly $\cop_i^t$ is the position of pursuer $\cop_i$.

The game is played as follows. First, the pursuers choose their initial positions $\cop_1^0, \cop_2^0, \ldots , \cop_{\ell}^0$. Next, the evader chooses his initial position $\rob^0$. Gameplay in round $t \geq 1$ proceeds as follows.
First, each pursuer $\cop_i$ moves  from its current position $\cop_i^{t-1}$ to a point $\cop_i^{t} \in B(\cop_i^{t-1},1)$. (Note that a pursuer may choose to remain stationary under these rules.) If $\cop_i^t = \rob^{t-1}$ for some $i \in \{ 1, \ldots \ell \}$, then the pursuers are victorious. Otherwise, the evader moves according to the same rule, moving from $\rob^{t-1}$ to a point $\rob^t \in B(r_{t-1}, 1)$. The evader wins if he evades capture forever. We consider the full-information version of this game, where each agent knows the environment, as well as the location of all other agents. Furthermore, the pursuers may coordinate their strategy.

 Pursuit-evasion in geometric environments has been extensively studied;  for a survey, see \cite{chi}. In addition to revealing interesting mathematical properties of polygons, the lion and man game with full visibility has practical applications in robotics. The pursuit strategy we provide can be used by two robots to intercept an intruder in a complex environment such as a casino. In such an environment, it might be possible to track the location of the intruder at all times using a camera network. However, the robots must still capture the intruder to detain him. Many other security, surveillance and search-and-rescue applications can be modeled similarly as pursuit-evasion games.
 
The  game's  history extends at least to  the 1930s when Rado posed the Lion-and-Man problem   in a circular arena, with lion chasing man \cite{littlewood}. Surprisingly, when time is continuous, man has a winning strategy \cite{littlewood,alonso}. However in the turn-based version our natural intuition prevails: the lion has a winning strategy. The turn-based version has received a good deal of attention (cf. \cite{alonso}, \cite{sgall}, \cite{KR}) and it is known that a single pursuer can always catch an evader in a simply connected polygon \cite{isler05tro}. Recently,  Bhadauria et al. \cite{bhadauria+klein+isler+suri} proved that 3 pursuers can always capture an evader in any polygonal environment (with obstacles). This 3-pursuer result is analogous to Aigner and Fromme's classic result about the pursuit-evasion game played a planar graph \cite{aigner+fromme}. When played on a graph, this game is known as \emph{cops and robbers}; see the surveys \cite{alspach, hahn} and the monograph \cite{bonato+nowakowski}. 
Variants of pursuit-evasion with limited pursuer sensing capabilities have also been studied \cite{guibas,isler04sidma} but we focus on the full-visibility case.

Herein, we study polygonal environments where one or two pursuers are sufficient for capture. Our results leverage two techniques: {guarding} and {projection}. Given a sub-environment $Q \subset P$, we say that a pursuer \emph{guards} $Q$ if (a) the evader is not currently in $Q$ and (b) if the evader crosses into $Q$ , the pursuer can respond by capturing him. A \emph{projection function} $\pi: P \rightarrow Q$ is a function that (a) is the identity map on $Q$ and (b) satisfies $d_Q(\pi(x), \pi(y)) \leq d_P(x,y)$. Note that the existence of a projection function $\pi: P \rightarrow Q$ guarantees that $Q$ is \emph{geodesically convex} in $P$, meaning that for any $x,y \in Q$, there is at least one shortest $(x,y)$-path $\Pi \subset Q$ (though there may be other shortest $(x,y)$-paths that leave $Q$).
We use projection functions to devise guarding strategies.  Our main result is the following.

\begin{theorem}
[Leapfrog Theorem]
\label{thm:leapfrog}
Suppose that $P$ contains a family of nested subregions $Q_0 \subset Q_1 \subset Q_2 \subset \cdots
\subset Q_k=P$ and, for $0 \leq i \leq k-1$, 
the following hold:
\begin{enumerate}
\item[(L1)] $Q_0$ is simply connected,
\item[(L2)] there is family of projections $\pi_i: Q_{i+1} \rightarrow Q_i$ 
\item[(L3)] $Q_{i+1} \setminus Q_{i}$ is a finite collection of simply-connected regions,
\item[(L4)] $Q_{i}$  intersects fewer obstacles than $Q_{i+1}$.
\end{enumerate}
Then $P$ is two-pursuer-win.
\end{theorem}

In particular, this theorem holds for a family of nested subsets that are geodesically convex in $P$. 

\begin{corollary}
\label{cor:leapfrog}
Suppose that $P$ contains a family of nested subregions $Q_0 \subset Q_1 \subset Q_2 \subset \cdots
\subset Q_k =P$ such that (L1), (L3), (L4)  hold for $0 \leq i \leq k-1$. If  
$Q_0 \cap \partial P$ contains at least two points and $Q_i$ is geodesically convex in $P$ for $0 \leq i \leq k$ then
 $P$ is two-pursuer-win.
\end{corollary}

\noindent
In proving Theorem \ref{thm:leapfrog}, we describe a \emph{leapfrog strategy} for two pursuers. First, $\cop_1$ evicts the evader from $Q_0$ and then guards this region. Next, the $\cop_2$ clears and guards $Q_1 \backslash Q_0$. In the process, $\cop_2$ ends up guarding all of $Q_1$, which frees up $\cop_1$ to leapfrog over $p_2$ to tackle $Q_2 \backslash Q_1$, and so on. Figure 
\ref{fig:leap} shows an environment and its leapfrog decomposition where each $Q_i$ is geodesically convex. This decomposition has some subtle features: region $Q_1$ intersects two more obstacles than $Q_0$, and $Q_3 \backslash Q_2$ is not connected. The leapfrog strategy handles both situations. In Section \ref{sec:leapfrog}, we will return to this example environment to illustrate the leapfrog strategy. 

\begin{figure}[ht]

\begin{center}
\begin{tikzpicture}[scale=1.25]

\leapenv

% shared
\draw[very thick] (1,0) -- (1, -.25);
\draw[very thick]  (.5, -1.5) -- (-1, -1.5) -- (-1.25, .2) -- (1.17,1.21);

\node at (-.85,-.5) {$Q_0$};
\node at (.4,.25) {$R_1$};
\node at (1.5,.4) {$R_2'$};
\node at (2,-1) {$R_2$};
\node at (3.6,-1) {$R_3$};

%\node[below] at (-.5, -1.5) {$u$};
%\node[below] at (.5, -1.5) {$v$};
%\node[below] at (2.5, -1.5) {$w$};

%\draw[fill] (.9,-.5) circle (1.5pt);
%\node[left] at (.9,-.5) {$\cop_2$};

%\draw[fill] (1.05,-.65) circle (1.5pt);
%\node[right] at (1.05,-.65) {$\rob$};

%\draw[fill] (-.5,.1) circle (1.5pt);
%\node[left] at (-.5,.1) {$\cop_1$};

\node at (4, .75) {$\begin{array}{l}   
Q_1 = Q_0 \cup R_1 \\  
Q_2 = Q_1 \cup R_2 \cup R_2' \\ 
Q_3 = Q_1 \cup R_3 \\    
\end{array}$};

\end{tikzpicture}
\end{center}

\caption{A  leapfrog partition $Q_0 \subset Q_1 \subset Q_2 \subset Q_3$ where each $Q_i$ is geodesically convex. Region $Q_1$ intersects two more obstacles than $Q_0$. Region  $Q_2 \backslash Q_1$ is disconnected and the two components of  $\partial Q_1 \cap \partial Q_2$ are thickly drawn. }
\label{fig:leap}
\end{figure}
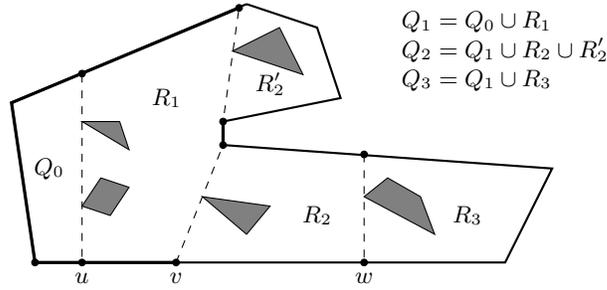

The leapfrog strategy completes in time $O( n \cdot h \cdot \mbox{diam}(P))$ where $n$ is the number of vertices of the environment (on both the outer boundary and on the obstacles), $h$ is the number of obstacles,  and $\mbox{diam}(P)$ is the diameter of the environment. 
(Note that the fourth condition ensures that $k \leq h+1$. )

Along the way, we also resolve the following question: when does one pursuer have a winning strategy a polygonal environment with a single obstacle $H$? The determining factor is the length of the perimeter of the convex hull of that obstacle.

\begin{theorem}
\label{thm:one-hole}
Suppose that $P$ is a polygonal environment with one obstacle $H$ with convex hull $J=\hull(H)$. Then $P$ is pursuer-win if and only if $J$ has perimeter $\ell \leq 2$. 
\end{theorem}

Finally, we complement this work with some computational results which identify sufficient (but not necessary) conditions for when an environment has a decomposition as described in Theorem \ref{thm:leapfrog}.  In Section \ref{sec:sweep}, we first make some observations about when the three-pursuer Minimal Path Strategy of Badauria et al.~\cite{bhadauria+klein+isler+suri} gives rise to a leapfrog decomposition. Such environments can be detected in $O(n^6)$ time. Next, we give an $O(n^2)$ algorithm for finding such a nested family of sets when the obstacles of $P$ are well-separated.  This result, formulated as  Theorem \ref{thm:2sweep} below,  requires technical definitions of \emph{sweepable} polygons and \emph{strictly two-sweepable} environments, so we defer its statement until later. This algorithm uses dual polygons and adapts the monotonicity results of Bose and van Kreveld \cite{bose}.

The paper is organized as follows. Section \ref{sec:lion} describes lion's strategy, a known strategy for a single pursuer to capture an evader in a simply connected environment. Section \ref{sec:proj} develops the projection and guarding framework that underlies our main results. In Section \ref{sec:obstacle}, we prove Theorem \ref{thm:one-hole} and in Section \ref{sec:leapfrog}, we prove Theorem \ref{thm:leapfrog}. Section \ref{sec:sweep} contains the  exploration of two subfamilies of leapfrog environments.  We conclude in Section \ref{sec:conc} with some avenues for future research.

%%%%%%%%%%%%% One Cop %%%%%%%%%%%%%%%

\section{Lion's strategy}
\label{sec:lion}

The lion-and-man game (cf.~\cite{alonso}),  takes place in a circular arena, with lion chasing man. Sgall \cite{sgall} considered the turn-based version played in the non-negative quadrant of the plane. He showed that a lion starting at $(x_0,y_0)$ captures a man starting at $(x_0', y_0')$ if and only if $x_0' < x_0$ and $y_0' < y_0$.  Kopparty and Ravishankar  \cite{KR} generalized this strategy to obtain the \emph{spheres strategy}  for  $\R^n$, where the evader $e$ is caught by pursuers $\cop_1, \cop_2, \ldots , \cop_{\ell}$  if and only if $e$ starts in the interior of the convex hull of the pursuers. During pursuit,  $\cop_i$  guards an expanding  circular region $B_i^t$ so that after step $t$, the pursuer has either captured the evader, or the area of the region $B_i^t$  he guards  is larger by a constant amount than the area of $B_i^{t-1}.$ 

Isler et. al. \cite{isler05tro} adapted this strategy for one pursuer in a simply connected polygon.  
We opt for the name ``lion's strategy'' for this scenario because the ball $B(x, r) = \{ y \in P \mid d(x,y) \leq r \}$ usually does not look like a traditional sphere due to the boundary edges.
Lion's strategy proceeds as follows. Suppose that the players start at points $\cop^0$ and $\rob^0$. The pursuer fixes a point $z \in P$ chosen so that $\cop^0$ is on the (unique) shortest path from $z$ to $\rob^0$. 
For convenience, we define $\rob^{-1} = \rob^0$.  The pursuer movement in round $t \geq 1$ is described in Algorithm \ref{alg:lion}.

\begin{algorithm}
\caption{Lion's Strategy}
\label{alg:lion}
\begin{algorithmic}
\item[] Given the positions $\rob^{t-1}$, $\rob^{t}$, $\cop^t$, such that $\cop^t$ is on the shortest path $\Pi^{t-1}$ from $z$ to $\rob^{t-1}$. 
\item[] To compute $\cop^{t+1}$:
\item[] Let $\Pi^t$ be the shortest path between $z$ and $\rob^{t}$.
\item[]  Choose $\cop^{t+1}$ on $\Pi^t$ such that $d(\cop^{t+1}, \rob^{t})$ is minimized, subject to $d(\cop^{t}, \cop^{t+1}) \leq 1$. 
\end{algorithmic}
\end{algorithm}

We observe that a pursuer using lion's strategy  actually  guards a monotonically increasing subset of the environment, namely $B(z,d^t)$ where $d^t=d(z,\cop^t)$. In other words, once a pursuer guards an area, he also prevents recontamination of that region for the remainder of the game.  The validity of Lion's Strategy in a simply connected environment is proven in \cite{isler05tro}. The proof shows that the pursuer can move from $\Pi^{t-1}$ to $\Pi^t$, and second that the distance between pursuer and evader decreases with this move.

\begin{lemma}[\cite{isler05tro}]
\label{lemma:lion}
Lion's strategy is a winning strategy for a single pursuer in a simply connected polygonal environment.
\end{lemma}

%%%%%%%%%%%%% Projection %%%%%%%%%%%%
% Projections.
\section{Projections and guarding}
\label{sec:proj}

Let $P$ be a polygonal environment with obstacles, and let $Q \subset P$ be a sub-environment. 
In this section, we define a broad class of projection functions from $P$ onto $Q$.
These projection functions  play a crucial role in our pursuit strategies. Such projections for pursuit-evasion appear in  recent results of Bhadauria et.~al.~\cite{bhadauria+klein+isler+suri}, and we draw heavily from their viewpoint.  
After giving a general definition of a projection, we will prove two pursuit results:  (1) a pursuer can evict the evader from $Q$ by chasing (and capturing) the evader's projection; and (2) a pursuer who is collocated with the evader's projection onto $Q$ can keep the evader from re-entering $Q$.

% PQ-projection definition.
\begin{mydef}
\label{def:proj}
A \emph{$(P,Q)$-projection}  is a function $\retract: P \rightarrow Q$ such that (1)
if $x \in Q$  then $\retract(x) = x$, and (2) for all $x,y \in P$, we have $d_{Q} (\retract(x), \retract(y)) \leq d_P(x,y)$.
\end{mydef}

\noindent
In other words, the $(P,Q)$-projection $\retract$ is the identity map on $Q$, and the mapping never increases the distances between points. Taking $x,y \in Q$, we find that $d_Q(x,y) \leq d_P(x,y)$, which means that $Q$ contains a shortest $(x,y)$-path (also known as a geodesic). In other words, if there is a projection function $\pi: P \rightarrow Q$ then $Q$ must be geodesically convex. 
%When $P$ and $Q$ are clear from the context, we refer to  $\retract$ as a \emph{projection}. 

%%%%%%%%%%%%
\begin{figure}[ht]
\begin{center}
\begin{tabular}{cccc}
\begin{tikzpicture}[scale=1]

 \clip (0,0) -- (3,.5) -- (4,3) -- (2.5,4) -- (.5,2) -- (-.5, 3.5)  -- cycle;

\draw[draw=none, fill=gray!45] (4,3) -- (2.5,1.5) -- (3,.5) -- cycle;

\draw[draw=none, fill=gray!45]  (-.5,1.5) -- (1,1) -- (1.5,0) -- (0,0) -- cycle;

\draw[draw=none, fill=gray!45] (0,3) -- (1.5,2.5) -- (3,4) -- (2.5,4) -- cycle;

\draw[very thick] (0,0) -- (3,.5) -- (4,3) -- (2.5,4) -- (.5,2) -- (-.5, 3.5)  -- cycle;

\draw[thick] (1,1) -- (2.5, 1.5) -- (1.5, 2.5) -- cycle;

\draw[dashed] (1,1) -- (1.5,0);
\draw[dashed] (2.5,1.5) -- (3,.5);

\draw[dashed] (2.5,1.5) -- (4,3);
\draw[dashed] (1.5,2.5) -- (3,4);

\draw[dashed] (1.5,2.5) -- (0,3);
\draw[dashed] (1,1) -- (-.5,1.5);

%%%%%%

\draw[fill] (-.1, 2.5) circle (2pt);
\draw[shorten <=4pt, -latex]  (-.1, 2.5) -- (.5,1.975) -- (1.25,1.75);
%\draw (.5,2) -- (-.5,2.33);
%\draw (.5,2) -- (-.5,2.33);

\draw[fill] (.25,.65) circle (2pt);
\draw[shorten <=4pt, -latex] (.25,.65) -- (1,1);

\draw[fill] (3.25,2.75) circle (2pt);
\draw[shorten <=4pt, -latex] (3.25,2.75) -- (2.25,1.75);

\node at (1.7, 1.7) {$Q$};

\end{tikzpicture}
& \qquad \qquad &
\begin{tikzpicture}[scale=.8]

\draw[thick] (0,0) -- (4,1) -- (6,3) -- (5,4) -- (0,4) -- (-2, 3) -- (-2,1) -- cycle;

\draw [very thick] (0,0) -- (-1,2) -- (0,4);

\node[above] at (.75,4) {$v = \pi_4=\pi_5$};
\node[below] at (0,0) {$u$};

\draw[fill] (0,4) circle (2pt) ;
\draw[fill] (0,0) circle (2pt);

\draw[thick, fill=gray] (-1,2) -- (2,1) -- (0,3) --cycle;

\draw[thick, fill=gray] (2,2.1) -- (3.25,2.75) -- (4,1.75) -- cycle;

\draw[dashed] (0,0) -- (2,1) -- (2,2.1) -- (3.1, 4);

%\draw (2,1) circle (1.1);
%\draw (2, 2.1) circle (1.1);
%\draw (2.6, 3) circle (1.1);

\draw[fill] (1, .5) circle (2pt) ;
\draw[fill] (-.5, 1) circle (2pt) ;

\node[above] at (1, .5) {$x_1$};
\node[left] at (-.5, 1) {$\pi_1$};

\draw[fill] (2, 1) circle (2pt) ;
\draw[fill] (-1, 2) circle (2pt) ;

\node[right] at (2, 1) {$x_2$};
\node[left] at (-1, 2) {$\pi_2$};

\draw[fill] (2, 2.1) circle (2pt) ;
\draw[fill] (-.5, 3) circle (2pt) ;

\node[left] at (2, 2.1) {$x_3$};
\node[left] at (-.5, 3) {$\pi_3$};

\draw[fill] (2.55, 3.05) circle (2pt);

\node[left] at (2.55, 3.05) {$x_4$};

\draw[fill] (3.1, 4) circle (2pt) ;

\node[above] at (3.1, 4) {$x_5$};

%\draw[fill] (-.55, 2.9) circle (2pt) ;

%\node[right] at (1,1.5) {$\Pi_2$};

%\node at (2.25,2.25) {$P_2$};

\end{tikzpicture}

\\
(a) && (b) 

\end{tabular}
\end{center}

\caption{(a) The metric projection onto convex subregion $Q$. The shaded areas map to the vertices of $Q$. (b) A path projection onto the minimal $(u,v)$ path, shown in bold. Each point $x_i$ projects to point $\pi_i$. }
\label{fig:projection-examples}
\end{figure}
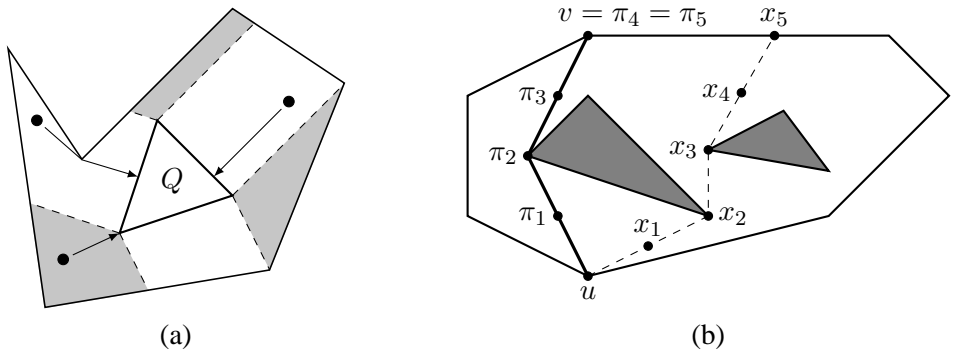

We now give a few examples of projections.
% Projection example: shortest-path projection to a convex set.
First, consider a simply connected polygonal environment $P$. Let $Q \subset P$ be a convex polygon.
If  $\partial P \cap \partial Q = \emptyset$,
then $P \backslash Q$ is a polygonal environment  with a single obstacle; otherwise the components of $P \backslash Q$ are all simply connected.
Define $\metricproj: P \rightarrow Q$ to be the mapping that takes $x \in P$ to the unique point $y \in Q$ such that $d_P(x,y) = \min_{z \in Q} d(x,z)$. 
% BA: reason for well-definedness of the projection.
Note that the convexity of $Q$ ensures that this mapping is well-defined.
Moreover, it is easy to see that $\metricproj$ is a $(P,Q)$-projection. 
Next, suppose that $P$ is not simply connected. Let $Q$ be a sub-environment with a  convex boundary such that every obstacle of $P$ is also contained in $Q$. In this case, the same function $\metricproj$ is still a $(P,Q)$-projection.
More broadly, when the sub-environment $Q$ is such that every $x \in P$ has a unique closest point in $Q$, we introduce the term \emph{metric projection}.  An example of a metric projection is shown in Figure \ref{fig:projection-examples}(a).

%% Def of metric projection.
\begin{mydef}
Suppose that $Q \subset P$ is such that for every $x \in P$ there is a unique point $y \in Q$ achieving $d_P(x,y) = \min_{z \in Q} d(x,z)$. The projection $\metricproj$ induced by this mapping is the {\bf metric projection} from $P$ onto $Q$.
\end{mydef}

%%%%%%%%%%%%
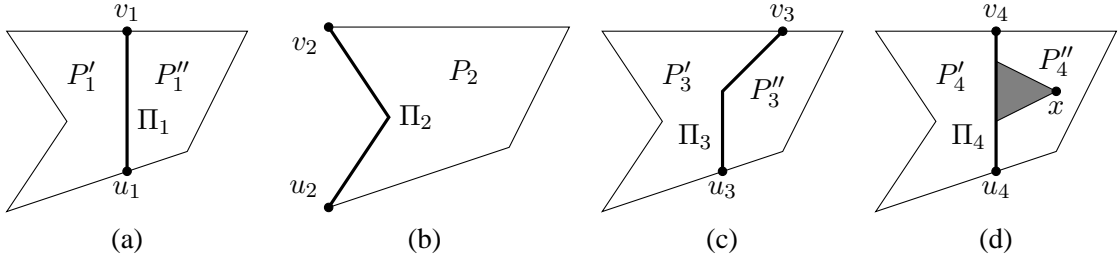
\begin{figure}[ht]
\begin{center}
\begin{tabular}{cccc}
\begin{tikzpicture}[scale=.8]

\draw(0,0) -- (3,1) -- (4,3) -- (0,3) -- (1,1.5) -- cycle;

\draw [very thick] (2,3) -- (2,.67);

\node[above] at (2,3) {$v_1$};
\node[below] at (2,.67) {$u_1$};

\draw[fill] (2,3) circle (2pt) ;
\draw[fill] (2,.67) circle (2pt);

\node[right] at (2,1.5) {$\Pi_1$};

\node at (1.25,2.25) {$P_1'$};
\node at (2.75,2.25) {$P_1''$};

\end{tikzpicture}
&
\begin{tikzpicture}[scale=.8]

\draw (0,0) -- (3,1) -- (4,3) -- (0,3) -- (1,1.5) -- cycle;

\draw [very thick] (0,0) -- (1,1.5) -- (0,3);

\node[below left] at (0,3) {$v_2$};
\node[above left] at (0,0) {$u_2$};

\draw[fill] (0,3) circle (2pt) ;
\draw[fill] (0,0) circle (2pt);

\node[right] at (1,1.5) {$\Pi_2$};

\node at (2.25,2.25) {$P_2$};

\end{tikzpicture}
&
\begin{tikzpicture}[scale=.8]

\draw (0,0) -- (3,1) -- (4,3) -- (0,3) -- (1,1.5) -- cycle;

\draw [very thick] (3,3) -- (2, 2) -- (2,.67);

\node[above] at (3,3) {$v_3$};
\node[below] at (2,.67) {$u_3$};

\draw[fill] (3,3) circle (2pt) ;
\draw[fill] (2,.67) circle (2pt);

\node[left] at (2,1.25) {$\Pi_3$};

\node at (1.25,2.25) {$P_3'$};
\node at (2.75,2) {$P_3''$};

\end{tikzpicture}
&
\begin{tikzpicture}[scale=.8]

\draw (0,0) -- (3,1) -- (4,3) -- (0,3) -- (1,1.5) -- cycle;

\draw [fill=gray] (2,2.5) -- (2,1.5) -- (3, 2) -- cycle;

\draw [very thick] (2,3) -- (2,.67);

\node[above] at (2,3) {$v_4$};
\node[below] at (2,.67) {$u_4$};

\draw[fill] (2,3) circle (2pt) ;
\draw[fill] (2,.67) circle (2pt);

\node[below] at (3,2) {$x$};
\draw[fill] (3,2) circle (2pt);

\node[left] at (2,1.25) {$\Pi_4$};

\node at (1.25,2.25) {$P_4'$};
\node at (3,2.5) {$P_4''$};

\end{tikzpicture}

\\
(a) & (b) & (c) & (d)

\end{tabular}
\end{center}

\caption{Some paths in polygonal  environments. There is a metric $(P, \Pi_i)$-projection only for $i=1,2$.
The paths $\Pi_1, \Pi_2, \Pi_4$ are minimal in their environments. The path $\Pi_3$ is only minimal for the sub-environment $P_3'$. There is no metric $(P_4, \Pi_4)$-projection since the closest point to $x$ is not unique.}
\label{fig:geomproj}
\end{figure}

%%%%%%%%%%%%% path projections.
Clearly, the metric projection $\metricproj$ is a $(P,Q)$-projection. However, there are many instances in which there is no well-defined metric projection because there are multiple nearest points; see Figure \ref{fig:geomproj} for some examples. Bhadauria, et al.~\cite{bhadauria+klein+isler+suri}  introduce a second type of projection that is less intuitive, but  applicable to a broader class of environments, including those with obstacles. 

% Minimal path def.
\begin{mydef}[Minimal Path \cite{bhadauria+klein+isler+suri}]
\label{def:minpath}
Suppose that $\Pi$ is a path in  environment $P$ dividing it into two sub-environments, and $P_e$ is the sub-environment containing the evader $e$. We say that $\Pi$ is  \emph{minimal} with respect to $P_e$ if, for all points $x,z \in \Pi$ and $y \in (P_e \backslash \Pi)$, we have $d_{\Pi}(x,z) \leq d_{P_e}(x,y) + d_{P_e}(y,z)$.
\end{mydef}

For example, a shortest path  between $u,v \in P$ is always minimal with respect to $P$.
In Figure \ref{fig:geomproj}, the paths $\Pi_1,\Pi_2, \Pi_4$ are minimal with respect to the whole environment. 
%The path  $\Pi_3$ is not minimal with respect to the whole environment, but it is minimal in the subenvironment $P_3'$.
There is no metric projection onto the minimal path $\Pi_4$: the obstacle results in the existence of two distinct points in $\Pi_4$ attaining the minimum distance to point $x$. We use the more robust  \emph{path projection} to deal with such an environment. An example of a path projection is shown in Figure \ref{fig:projection-examples}(b).

% path project def.
\begin{mydef}[Path Projection \cite{bhadauria+klein+isler+suri}]
\label{def:path-proj}
Let $u, v \in \partial P$ and let
 $\Pi_{u,v}$ be a minimal $(u,v)$-path in $P_e$.
 For $x \in P_e$ with $d(u,x) \leq d(u,v)$, define $\phi(x)$ to be the point on $\Pi_{u,v}$ at distance $d(u,x)$ from $u$. When $d(u,x) > d(u,v)$ define $\phi(x) =  v$. The mapping $\phi$ is called the \emph{path projection} of $P_e$ onto $\Pi_{u,v}$. Setting $Q$ to be the complement of ${P_e}$, we extend this to a projection $\phi: P \rightarrow Q$ by setting $\phi(x) = x$ for $x \in Q$. 
\end{mydef}

Note that we restrict $u,v$ to lie on the boundary of $P$. The proof that this mapping is a $(P,Q)$-projection is given in \cite{bhadauria+klein+isler+suri}. Considering Figure \ref{fig:geomproj}, we see that $\Pi_3$ is a minimal path in $P_3'$, but not in $P_3''$. Meanwhile, there are  path projections from each of $P_4',P_4''$ to $\Pi_4$. 
%% Lemma: pursuers can move between projections of evader's position.
Restricting a pursuer's  movement to the evader's projection is a key component of the  pursuit strategies developed in the following sections.
Once a pursuer captures the evader's projection, the pursuer can maintain that colocation after every pursuer turn thereafter. Indeed, if the evader moves from $e$ to $e'$ then  $1 \geq d (e, e') \geq d (\pi(e), \pi(e'))$. We state this as a useful lemma.

\begin{lemma}
\label{lemma:shadow}
Let $Q \subset P$ be a sub environment with a projection $\pi: P \rightarrow Q$. Suppose that the pursuer $p$ starts at $\pi(e)$. After the evader moves from $e$ to $e'$, the pursuer can  move from $\pi(e)$ to $\pi(e')$. \proofend
\end{lemma}

%% Extending/Patching together projections.
During our pursuit, we frequently  divide the environment into intersecting regions. The following lemma explains how to  patch together projections of overlapping subregions.

\begin{lemma}
\label{lemma:piecewise}
Let $P$ be a polygonal environment with sub-environments $Q, P_1, P_2$ such that  $P=P_1 \cup P_2$ and $Q \subset P_1 \cap P_2$. Suppose that  for $i=1,2$, we have  a  $(P_i, Q)$-projection $\retract_i$ with  $\retract_1(x) = \retract_2(x)$ for every $x \in P_i \cap P_2$.  Then the function  
\begin{displaymath}
   \retract(x) = \left\{
     \begin{array}{ll}
      x & x \in Q \\
       \retract_{1}(x) &  x \in P_{1} \backslash Q \\
       \retract_{2}(x) & x \in P_{2} \backslash   P_1
     \end{array}
   \right.
\end{displaymath}
is a projection from $P$ to $Q$.
\end{lemma}

\Proof 
For  points  $x \in P_1$ and $y \in P = P_1 \cup P_2$, let $\Pi$ represent a minimal path from $x$ to $y$ in $P$. 
Partition $\Pi$ into a finite collection of subpaths $\Pi = \{\Pi_1, \Pi_2, \ldots, \Pi_k \}$ where the odd indexed paths are in $P_1$ and the even indexed paths are in $P_2$. Let $u_{i-1}, u_i$ be the endpoints of $\Pi_i$, so that $u_0=x$ and $u_k=y$.  We consider the case that $k$ is even; the proof for odd $k$ is similar. We have
\begin{eqnarray*}
d_{Q} (\retract(x), \retract(y)) &\leq& d_{Q}(\retract_1(u_0), \retract_1(u_1)) +
d_{Q}(\retract_2(u_1), \retract_2(u_2)) + \cdots + d_{Q}(\retract_k(u_{k-1}), \retract_k(u_k)) \\
&\leq& d_{P_1}(u_0,u_1) + d_{P_2}(u_1,u_2) + \cdots + d_{P_k}(u_{k-1},v_k) 
= d_{P}(x,y). \qquad \Box
\end{eqnarray*}

By induction,  the analogous result holds for any finite collection of projections, with pairwise agreement on common intersections.

\begin{cor}
% BA: added label.
\label{cor: piecewise projection}
Let $P=P_1 \cup P_2 \cup \cdots P_k$ be a polygonal environment with    a sub-environment $P_0$ such that for all $1 \leq i \leq k$, we have $P_0 \subset  P_i$. Suppose that there exists a family of $(P_i,P_0)$-projections $\retract_i$ such that  for any $1 \leq i,j \leq k$, we have $\retract_i(x) = \retract_j(x)$ for every $x \in P_i \cap P_j$. Then the piecewise  function  
\begin{displaymath}
   \retract(x) = \left\{
     \begin{array}{ll}
      x & x \in P_0 \\
       \retract_{i}(x) &  x \in P_{i} \backslash (P_0 \cup \cdots \cup P_{i-1}), 1 \leq i \leq k 
     \end{array}
   \right.
\end{displaymath}
is a projection from $P$ to $P_0$. \qed
\end{cor}

We use Lemma \ref{lemma:piecewise}  to extend our definition of projection to apply to a  loop $\Lambda$ (that is, closed path) that intersects the boundary of $P$ in at least two points $u, v$.  
%First, we observe that Definition \ref{def:minpath} remains valid when $\Pi$ is a loop (with $u=v$). However, the path projection becomes problematic, since the evader can wind around $\Pi$. We resolve this by splitting the loop into two $(u,v)$-paths $\Lambda = \Pi_1 \cup \Pi_2$. 
%This naturally partitions $P$ into three sub environments: the interior of $\Lambda$, and the two exterior environments $P_1$ and $P_2$ bounded by  $\Pi_1$ and $\Pi_2$ respectively.

% Minimal path def.
\begin{mydef}[Loop, Minimal Loop]
Let $u,v \in \partial P$. 
A \emph{loop} $\Lambda$  consists of two internally disjoint $(u,v)$-paths $\Pi_1, \Pi_2$. These paths divide the environment $P$ into three sub-environments: the interior $Q$ between the two paths, and exterior environments $P_1, P_2$, bounded by $\Pi_1, \Pi_2$ respectively. The loop $\Lambda$ is $(u,v)$-\emph{minimal} when  $\Pi_1$ is a minimal path for $P_1$ and  $\Pi_1$ is a minimal path for $P_2$.
\end{mydef}

% that intersects $\partial P$ in points $u$, $v$, creating paths $\Pi_1$ and $\Pi_2$
%of length $\ell$ that contains boundary point $u$ in  environment $P$ . Parameterize this loop as $\Pi(t)$, $0 \leq t \leq \ell$ so that $\Pi(0) = \Pi(\ell) = u$.  Suppose that the loop $\Pi$ divides $P$ into two sub-environments, and $P_e$ is the sub-environment containing the evader $e$. We say that the loop $\Pi$ is  \emph{minimal} with respect to $P_e$ if for every $y \in P_e$, there exists an $\epsilon > 0$ such that for  all $0 < \delta < \epsilon$, we have $\ell - 2 \delta \leq d_{P_e} (\Pi(\delta), y) + d_{P_e}(y, \Pi(\ell - \delta))$. 

%%%%%%%%%%%%
\begin{figure}[ht]
\begin{center}
\begin{tabular}{ccccc}
\begin{tikzpicture}[scale=.8]

\draw (0,0) -- (5,0) -- (4,3) -- (0,3) -- (.5,1.5) -- cycle;

\draw [fill=gray] (1.25,2.25) -- (1.25,.75) -- (3, 1.5) -- cycle;

\draw [very thick] (2,3) -- (1.25,2.25) -- (1.25, .75)  -- (2, 0)  -- (3,1.5)--cycle;

\node at (.9,1.25) {$\Pi_1$};
\node at (3.1,2) {$\Pi_2$};

\node at (.75,2.25) {$P_1$};
\node at (3.75,.75) {$P_2$};

\draw[fill] (2,3) circle (2pt) ;
\draw[fill] (2,0) circle (2pt) ;

\node[above] at (2,3) {$u$} ;
\node[below] at (2,0) {$v$} ;

\end{tikzpicture}
&

%%%%%%
\begin{tikzpicture}[scale=.8]

\draw (0,0) -- (5,0) -- (4,3) -- (0,3) -- (.5,1.5) -- cycle;

\draw [fill=gray] (1.25,2.25) -- (1.25,.75) -- (3, 1.5) -- cycle;

\draw [very thick] (2,3) -- (1.25,2.25) -- (.5, 1.5) -- (1.25, .75)  -- (2, 0)  -- (3,1.5)--cycle;

\node at (1,.5) {$\Pi_1$};
\node at (3.1,2) {$\Pi_2$};

\node at (-.5,1.5) {$P_1$};
\node at (3.75,.75) {$P_2$};

\draw[fill] (2,3) circle (2pt) ;
\draw[fill] (2,0) circle (2pt) ;

\node[above] at (2,3) {$u$} ;
\node[below] at (2,0) {$v$} ;

\draw [-latex] (-.35, 1.65) -- (.5, 2.5);
\draw [-latex] (-.25, 1.25) -- (.5, .5);

\end{tikzpicture}

&
\begin{tikzpicture}[scale=.8]

\draw (0,0) -- (5,0) -- (4,3) -- (0,3) -- (.5,1.5) -- cycle;

\draw [fill=gray] (1.25,2.25) -- (1.25,.75) -- (3, 1.5) -- cycle;

\draw [very thick] (.5, 1.5) -- (1.25,2.25)   -- (3,1.5)  -- (1.25, .75) -- (0,0) --cycle;

\node at (-.75,.8) {$ P_1 = \Pi_1$};
\node at (2.5,2.25) {$\Pi_2$};

\node at (3.75,.75) {$P_2$};

\draw[fill] (.5,1.5) circle (2pt) ;
\draw[fill] (0,0) circle (2pt) ;

\node[left] at (.5,1.5) {$u$} ;
\node[below] at (0,0) {$v$} ;

\end{tikzpicture}

\\
(a) & (b) & (c)
\end{tabular}
\end{center}

\caption{Minimal loops in a polygonal  environment. 
For each loop,  $\Pi_1$  is a minimal path in  in $P_1$ and $\Pi_2$ is a minimal path in $P_2$. In (b), the sub-environment $P_1$ is disconnected. In (c), we have $P_1 = \Pi_1$ since the path $\Pi_1$ is part of the external boundary of the environment. }
\label{fig:loopproj}

\end{figure}
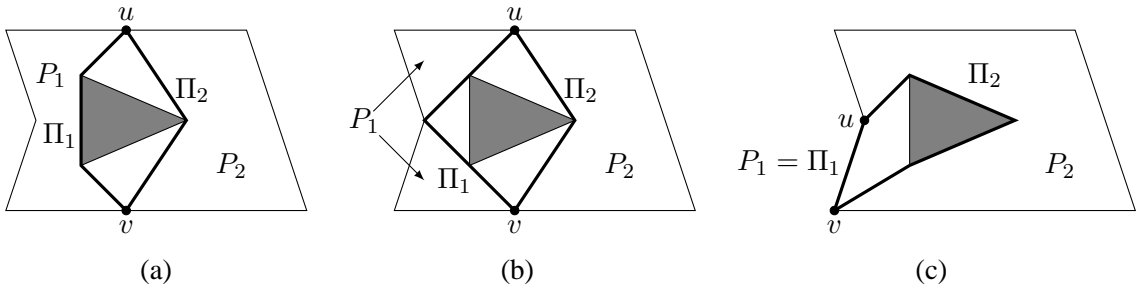

\begin{lemma}
\label{lemma:minloop}
Let $\Lambda$ be a $(u,v)$-minimal loop in polygonal environment $P$ where $u,v \in \partial P$.  Let  $Q$ be the sub-environment bounded by the loop  $\Lambda$. There is a projection $\phi: P \rightarrow Q$.
\end{lemma}

\begin{proof}

For $i=1,2$, let $P_i$ be the sub-environment of $P$ bounded by the minimal path $\Pi_i$, and let  $\phi_i : P_i \cup Q \rightarrow Q$ be the path projection. By Lemma \ref{lemma:piecewise}, we can combine these projections to get a projection $\phi: P \rightarrow Q$.
\end{proof}

%%%%%%%%
%%%%%%%%
%%%%%%%%

%-----------------------------------------------------------------------------------------------------------------------------------------------------------------------
%% Guarding Section.
%\section{Guarding regions using projection}
%-----------------------------------------------------------------------------------------------------------------------------------------------------------------------
In the remainder of this section, we explain how to use projections to evict the evader from a region, and then guard that region thereafter (that is, prevent the evader from re-entering). First, we show that if the pursuer is co-located with the evader's projection onto a minimal path, then the pursuer can prevent the evader from crossing it.

% Guarding lemma.
\begin{lemma}[Guarding Lemma]
\label{lemma:guard}
Let $Q \subset P$ with the projection $\retract: P \rightarrow Q$. Suppose that $\cop^1 = \retract(\rob^0)$.  Then the pursuer can maintain $\cop^{i+1} = \retract(\rob^{i})$ for $i \geq 1$. Furthermore, if the evader moves so that a shortest path from $\rob^{t-1}$ to $\rob^t$ intersects $Q$, then $d_P(\cop^t,\rob^t) \leq d_P(\rob^{t-1},\rob^t) \leq 1$, so the pursuer can capture the evader at time $t+1$.
\end{lemma}

\begin{proof}
The first claim follows easily by induction, since 
$$d_Q(\cop^i, \retract(\rob^{i})) = d_Q(\retract(\rob^{i-1}), \retract(\rob^{i})) \leq d_P(\rob^{i-1}, \rob^{i}) \leq 1,$$
meaning that the pursuer can remain on the projection of the evader.
For the second claim, suppose that the shortest path from $\rob^{t-1}$ to $\rob^t$ includes the point $x \in Q$. Then 
$$
d_P(\cop^t, \rob^t) \leq d_Q(\cop^t,x) + d_P(x, \rob^t) \leq d_P(\rob^{t-1},x) + d_P(x,\rob^t) = d_P(\rob^{t-1},\rob^t) \leq 1.
$$
Therefore the pursuer can move to $\rob^t$ on his next turn.
 \end{proof}

% Definition of guarding with respect to a projection.
When a pursuer $\cop$ follows the strategy in Lemma \ref{lemma:guard}, we say that $\cop$  \emph{guards} $Q$ with respect to the projection $\retract$.
Note that we must have $\cop^t=\retract(\rob^{t-1})$ before the pursuer can start to guard $Q$. After that, the pursuer never leaves the region he guards.  If the evader travels into  $Q$ (or across $Q$ in a single turn), then the pursuer apprehends the evader on his next turn. One example of this last kind of capture is given in the following corollary.

% Lemma: evader is caught if he enters a guarded region.
\begin{cor}
\label{cor:guard}
Suppose that pursuer $\cop$ is guarding subregion $Q \subset P$ with respect to $\retract: P \rightarrow Q$. Suppose further that $P \backslash Q$ is the disjoint union of simply connected components $R_1, R_2, \ldots, R_s$. If the evader moves from $R_i$ to $R_j$, where $i \neq j$, then the pursuer can catch him on his next turn. \proofend
\end{cor}

%% Clearing a polygon.
Our final pursuit lemma asserts that a single pursuer can evict the evader from any simply connected subregion with a valid projection function. We refer to this as \emph{clearing the subregion}.

% Clearing lemma.
\begin{lemma}
[Clearing Lemma]
\label{lemma:clear}
Let $Q$ be a simply connected subenvironment  of $P$ with  
projection
$\retract: P \rightarrow Q$.
In a finite number of moves, one pursuer can either capture the evader or guard $Q$.
\end{lemma}

\begin{proof}
The pursuer executes lion's strategy in the simply connected region $Q$, chasing after $\retract(\rob)$. By Lemma \ref{lemma:lion},  a single pursuer can capture an evader in a simply connected region in finite time, say  $\cop^t = \retract(\rob^{t-1})$. 
If $\rob^{t-1} = \retract(\rob^{t-1})$, then the evader is caught. Otherwise, the pursuer has attained position to guard $Q$ from this time forward, as described by Lemma \ref{lemma:guard}.
\end{proof}

%%%%%%%%%%%%% One Obstacle %%%%%%%%%%%%%%%
%-----------------------------------------------------------------------------------------------------------------------------------------------------------------------
% evader strategy for environments with one large obstacle.
\section{Environments with one obstacle}
\label{sec:obstacle}
%-----------------------------------------------------------------------------------------------------------------------------------------------------------------------

We use the projection framework from Section \ref{sec:proj} to prove Theorem \ref{thm:one-hole}.
Two pursuers are always enough to catch an evader when there is a single obstacle in the environment. Indeed, $\cop_1$ can move to guard a shortest path from the obstacle to the boundary of $P$. This  makes the environment simply connected, so $\cop_2$ can catch $\rob$ using lion's strategy. However, if the obstacle $H$ is small enough, one pursuer can actually catch the evader. 
The critical factor is the length of the boundary of the convex hull $J=\hull(H)$. Theorem \ref{thm:one-hole} states that the  environment is one-pursuer-win if and only if this boundary length is smaller than 2. The theorem follows directly from Lemmas \ref{lemma:robberwin} and \ref{lemma:copwin}  below.

% Lemma: evader wins if the obstacle is big enough.
\begin{lemma}
\label{lemma:robberwin}
Suppose that $P$ is a polygonal environment with one obstacle $H$ whose convex hull $J=\hull(H)$ has perimeter $\ell > 2$. The evader has a winning strategy against a single pursuer.
\end{lemma}

\begin{proof}
For simplicity of exposition, we assume that $\hull(H)$ does not intersect the boundary of the environment. The proof can be adapted for this case,  but we must redefine the convex hull of the obstacle in the natural way to handle the interaction with the external boundary.
%% Convex Case.
We first consider the case where the obstacle $H$ is convex. Let $\rho: P \rightarrow H$ be the metric projection. The mapping $\rho$ projects every point in $P \backslash H$ onto  $\partial H$.
We prove that the evader can always guarantee that  $d(\cop^i, \rob^i)  > d(\rho(\cop^i), \rob^i) = \ell /2 > 1$ for $i \geq 0$, which   means that the pursuer can never catch the evader.

% Robbers strategy.
The game begins when the pursuer chooses his location $\cop^0$ in $P$. The evader responds by placing himself at the unique point $\rob^0$ on $\partial H$ that is distance $\ell/2$ from $\rho(\cop^0)$. 
Proceeding by induction, assume that $d(\rho(\cop^i),\rob^i) = \ell/2$. The pursuer moves to a new location $\cop^{i+1}$ with $1 \geq d(\cop^i, \cop^{i+1}) \geq d(\rho(\cop^i), \rho(\cop^{i+1}))$. The evader responds by moving from $\rob^i$ to the unique point on the perimeter of $H$ that is distance $\ell/2$ from $\rho(\cop^{i+1})$. Of course, $d(\rob^i, \rob^{i+1}) = d(\rho(\cop^i), \rho(\cop^{i+1})) \leq d(\cop^i, \cop^{i+1}) \leq 1$, so the evader can attain this position, evading capture.

%% Non-convex case.
Next, we consider the case where the obstacle $H$ is not convex, but has convex hull  $J = \hull(H)$ with perimeter $\ell > 2$. 
Analogous to the above case, the evader will restrict his movement to the perimeter $\partial J$ of the convex region $J$.  As long as the pursuer does not enter $J$, the argument for convex obstacles shows that the evader can remain at distance $\ell/2$ from the pursuer projection. Meanwhile, entering $J$ (or more precisely, some component of $J \backslash H$) is worse for the pursuer than staying on the boundary of $J$. We make this more precise by using projections.

% The piecewise projection.
Let $\rho: P \rightarrow J$ be the metric projection. Next, we define projections for the areas in $J$. Let the vertices of $J$ be $v_0, v_1, \ldots, v_{k-1}$, indexed counterclockwise. Let $\Pi_i$ be the line segment joining $v_i, v_{i+1}$ (here the index is modulo $k$).
 Let $Q_i$ be the simply connected component of $J \backslash H$ whose boundary includes $\Pi_i$.
  (Note that $Q_i = \Pi_i$  when the segment between $v_i, v_{i+1}$ is  part of obstacle $H$.)
  Let $\phi_i: Q_i \rightarrow \Pi_i$ be the path projection from $Q_i$ to $\Pi_i$, anchored at $v_i$. Finally, we define the piecewise function $f: P \rightarrow \partial J$ as
\[
\pi(x) = \left\{
\begin{array}{cc}
x & x \in \partial J, \\
\rho(x) & x \in P \backslash J, \\
\phi_i(x) & x \in Q_i \backslash \Pi_i.
\end{array}
\right.
\]
This piecewise function is a projection from $P$ to $\partial J$ by Corollary~\ref{cor: piecewise projection}.
Suppose that $\cop^t \in P \backslash J$ and $\cop^{t+1} \in Q_i$. Let $x \in \Pi_i$ be on a minimal $(\cop^t,\cop^{t+1})$-path. Then
\[
1 \geq d(\cop^t, \cop^{t+1}) = d(\cop^t, x) + d(x, \cop^{t+1}) \geq d(\pi(\cop^t),x) + d(x, \pi(\cop^{t+1})) = d(\pi(\cop^t), \pi(\cop^{t+1})),
\]
where the last equality holds because  $x$ is on  a shortest path from $\pi(\cop^t)$ to $\pi(\cop^{t+1})$ in $\partial J$.
Therefore, the evader can move to the point at distance $\ell/2$ from $\pi(\cop^{t+1})$. The analogous argument holds when the pursuer moves from $Q_i$ to $P \backslash J$, or from $Q_i$ to $Q_j$.
\end{proof}

%%%%%%%%%%%%%%%%%
We now turn to the pursuer-win situation.
The proof of this result rests mainly on the following lemma, which shows that we can transition from guarding a line segment $\Pi$  to lion's strategy in such a way that the evader cannot cross $\Pi$ without being caught.

\begin{lemma}
\label{lemma:transition}
Let $\Pi$ be a line segment in $P$ that connects boundary points $u,v \in \partial P$. Let $P_e$ be the sub-environment of $P \backslash \Pi$ containing the evader, and let $Q  = P \backslash P_e$. Let $\rho : P_e \rightarrow Q$ be the metric projection. If the pursuer starts at $\cop^0= \rho (\rob^0)$ then 
 in a single move, the pursuer can transition to lion's strategy, keeping the line segment $\Pi$ within his guarded region. 
\end{lemma}

\begin{proof}
Let $\Lambda$ be the line through $p^0$ and $e^0$.
Let  $a = d(\rob^1, \Lambda)$ and  let 
$$b=\max \left\{ 2, d(\cop^0, u), d(\cop^0, v), \sqrt{a} \,  \right\}.$$ 
We use the coordinate system with $\cop^0=(0,0)$ and where $\Lambda$ is the  $x$-axis, so that
$\rob^0 = (a_0, 0)$ and $\rob^1=(a,h)$ for some $a_0, a \in \R^+$ where $(a-a_0)^2 + h^2 \leq 1$. 
 We can then safely replace $\Pi$ by a line segment of length $2b$ (which extends outside the boundary of $P$) with endpoints
 $u = (0 -b)$, $v= (0,b)$. The relevant geometry is shown in Figure \ref{fig:dist1}. We will show that
in a single move, the pursuer can transition to a lion's strategy that still guards $\Pi$. In particular, we will find $s \in \R^+$ such that the pursuer can move onto the circle with center 
$z=(-s, 0)$ and radius $ \sqrt{ s^2 + b^2}$, which contains  the segment $\Pi$. We will see that choosing $s=b^4/a$ is sufficient.

%%%%%%%%%%%%%%%%

\begin{figure}[ht]
\begin{center}
\begin{tikzpicture}[scale=1]

\path (0,0) coordinate (C0);
\path (1,0) coordinate (R0);
\path (.9,.9) coordinate (R1);
\path (0,1) coordinate (R1P);

\path (0,2) coordinate (B);
\path (0,-2) coordinate (BB);

\path (-6.75,0) coordinate (X);

\draw[fill] (C0) circle (1pt);
\draw[fill] (R0) circle (1pt);
\draw[fill] (R1) circle (1pt);

\draw[fill] (B) circle (1pt);

\draw[fill] (X) circle (1pt);

\node[below left=2pt] at (C0) {$\cop^0=(0,0)$};
\node[right=2pt] at (R0) {$\rob^0=(a_0,0)$};
\node[right=2pt] at (R1) {$\rob^1=(a,h)$};

\node[above=2pt] at  (B) {$(0,b)$};

\node[below=2pt] at (X) {$z=(0,-s)$};

\draw (0,-1) -- (B);

\node[left] at (0,-1) {$\Pi$};

\node at (-3.5,-.25) {$\Lambda$};

\draw (C0) -- (R0);

\draw (B) -- (X) -- (C0);

\draw (X) -- (R1);

\begin{scope}[shift={(B)}]
\draw[dashed] (0,0) arc (17.25: -9: 6.75);
\end{scope}

\node at (-3.5,1.4) {$\sqrt{s^2 + b^2}$};

\path (.25,.83) coordinate (Z);
\path (.25,0) coordinate (Z1);
\path (0,.83) coordinate (Z2);

\draw[fill] (Z) circle (1pt);

\node at (1.25,1.5) {$\cop^1=(x, y)$};

\draw[-latex] (.7, 1.35) -- (.35, 1);

%\draw (Z1) -- (Z) -- (Z2);

\draw[dotted] (C0) -- (Z);
\draw[dotted] (R0) -- (R1);

\end{tikzpicture}

\caption{The transition from guarding to lion when $\rob^1$ is distance $a$ away from the guarded line $\Pi$, and distance $h$ from the line through $\cop^0$ and $\rob^0$.}

\label{fig:dist1}
\end{center}
\end{figure}

The center of our circle $z=(-s,0)$ and the pursuer location $\cop^1 = (x,y)$ must satisfy 
\begin{eqnarray}
\label{eqn:2one}
x^2 + y^2 & \leq& 1, \\
\label{eqn:2two}
(s+x)^2 + y^2 &=& s^2 + b^2, \\
\label{eqn:2three}
\frac{y}{s+x} &=& \frac{h}{s+a}.
\end{eqnarray}
The last equation is equivalent to $(s+x) = y(s+a)/h$.
Using this value in equation \eqref{eqn:2two} and setting $s=b^4/a$ yields
\[
y^2 \leq h^2 \left( \frac{b^8/a^2 + b^2}{b^8/a^2 + 2  b^4 + a^2 + h^2}   \right)<  h^2 \leq 1.
\]
We have
\[
x = \frac{s+a}{h}y - s \leq \sqrt{s^2+b^2}  - s\leq  s \left( 1 + \frac{b^2}{2s^2} \right)  - s = \frac{b^2}{2s}  = \frac{a}{2b^2} < 1,
\]
because we chose $b$ so that $a \leq b^2$.
Therefore
\[
x^2 + y^2 
\leq
\frac{b^4}{4s^2} + \frac{h^2 (s^2 + b^2)}{s^2 + 2 a s + a^2 + h^2} 
\leq 
\frac{b^4}{4s^2} + \frac{ s^2 + b^2}{s^2 + 2 a s + a^2} 
< 1.
\]
In other words, this pair $(x,y)$ satisfies the final constraint equation \eqref{eqn:2one}. This also guarantees that $x < a$ because $y < h$ and equation \eqref{eqn:2three} holds. 
In conclusion, the pursuer can take one step onto the line connecting $z$ and $\rob^1$ and immediately guard a disc that contains the line segment $\Pi$. 
\end{proof}

We note that the progress that the pursuer makes during this transition depends upon how close the evader is to the obstacle. Indeed, we use the point $z$ as our center, which is at distance $s=b^4/a$, so the capture time is inversely proportional to the evader's initial distance from the obstacle. Given our restrictive definition of capture (colocation, as opposed to proximity), there is no way around this. This problem does not manifest itself in the two-cop strategy in Theorem \ref{thm:leapfrog} since the pursuers alternate their guarding and pursuit roles.

\begin{lemma}
\label{lemma:copwin}
Suppose that $P$ is a polygonal environment with one obstacle $H$ whose convex hull $J=\hull(H)$ has perimeter $\ell \leq 2$. Then $P$ is one-pursuer-win. 
\end{lemma}

\begin{proof}
Initially, the pursuer chooses his position $\cop^0$ to be some point on the boundary of $J =\hull(H)$. The evader then chooses his initial position $\rob^0$. On his first turn, the pursuer moves to the metric projection $\cop^1 = \rho(\rob^0)$ on $\partial J$. Draw the  maximal line segment  $\Pi \subset P$ through $\cop^1$ that is perpendicular to the segment joining $\cop^1$ and $\rob^0$. Let the endpoints of $\Pi$ be $u,v \in \partial P$.
%  (If $\cop^1$ is not on a vertex of $J$, then $\Pi$ will include a line segment of $\partial J$.)   
The pursuer currently guards $\Pi$. If the evader is not in $J$, then the area guarded by the pursuer contains the entire obstacle $H$, which means that $P_e$ is simply connected. If the evader is in $J$, then he is trapped in a simply connected area between $\Pi$ and $H$.
In either case, Lemma \ref{lemma:transition} allows  the pursuer can transition from guarding $\Pi$ to lion's strategy in a simply connected environment.  Using Lion's Strategy, the pursuer catches the evader by Lemma \ref{lemma:lion}.
\end{proof}

%%%%%%%%%%%%%%%%%%%%%

%%%%%%%%%%%%% Leapfrog %%%%%%%%%%%%

%-----------------------------------------------------------------------------------------------------------------------------------------------------------------------
%% The Leapfrog Strategy.
%-----------------------------------------------------------------------------------------------------------------------------------------------------------------------
\section{The leapfrog strategy}
\label{sec:leapfrog}

We use the projection framework of Section \ref{sec:proj} to prove our main result, Theorem \ref{thm:leapfrog}. We show that if an environment has a decomposition satisfying (L1) -- (L4) then two pursuers can use the \emph{leapfrog strategy} to capture the evader. Figure \ref{fig:leap} showed an example of an environment with a leapfrog partition $Q_0 \subset Q_1 \subset Q_2 \subset Q_3=P$. Figure
\ref{fig:leapfrog-example} summarizes  a leapfrog pursuit  in this environment, giving an illustrative example of the  strategy described in the proof of Theorem \ref{thm:leapfrog}. Each subfigure shows the minimal path used to find the location of the evader projection $\sigma_i$ on $\partial Q_i$. Subfigure (c) also shows the shortest path used for lion's strategy, and these two paths coincide in subfigure (f).

Intuitively, this strategy works as follows. While the first pursuer guards a subregion $Q$, the second pursuer clears new territory and guards a larger subregion $Q'$ containing $Q$. At that point, the second pursuer switches into guarding mode, and the first pursuer leapfrogs over him to clear new territory. This process continues until the evader is caught. 
However, there is another more subtle way for the pursuers to make progress: if the evader ``makes a mistake'' by passing through $Q$, then the current guarding pursuer can immediately capture the evader in his responding move.
For example, in  Figure \ref{fig:leapfrog-example} (d), the evader cannot move between connected components of $Q \backslash Q_1$ without being caught by $\cop_2$.
Finally, we 
note that the leapfrogging means that our current argument is completely independent of the material in Section \ref{sec:obstacle} since neither pursuer must alternate directly from guarding to pursuing in a single move.

Before we begin the proof, we reflect on the conditions (L1) -- (L4).
Conditions (L1) and (L2) let us use the projection framework in Section \ref{sec:proj} for region $Q_0$.   Condition (L3) ensures that  $\partial Q_{i+1} \cap \partial Q_i \neq \emptyset.$ Indeed, these boundaries are both polygons, and if they are disjoint, then the closed path $\delta Q_{i+1}$ is a nontrivial loop in $Q_{i+1} \backslash Q_i$.
Condition (L4) is included for efficiency: if we start with a larger family of nested regions, then we can simply ignore the subregions that do not touch additional obstacles; such a coarsening is shown in Figure \ref{fig:twosweeppartition} below.

\begin{figure}[ht]
\begin{center}
\begin{tabular}{ccc}

%%%%% one
\begin{tikzpicture}
\leapenv

\draw[fill] (-1, -1.5)  circle (1.5pt);
\node[left] at (-1, -1.5)  {$\cop_1$};

\draw[fill] (-.85,-.4) circle (1.5pt);
\node[above] at (-.85,-.4) {$\cop_2$};

\draw[fill] (1.29, .18) circle (1.5pt);
\node[above] at (1.29, .18) {$\rob$};

\draw (.5,-1.5) -- (1,-.25) --  (1,0) -- (1.29, .18);
\draw (-.5,-1.5) -- (1.05, .36);

\draw[thick, fill=gray] (1.05, .34) circle (1.5pt);
\node[left] at (1.05, .34) {$\sigma_1$};

\draw[thick, fill=gray] (-.5, .51) circle (1.5pt);
\node[above] at (-.5, .51) {$\sigma_0$};

%\draw (1,0) circle (.133in);
%\draw (-.5, -1.5) circle (.95in);

\node at (1,-2) {(a)};

\end{tikzpicture}

&
\qquad \qquad
&

%%% two
\begin{tikzpicture}

\draw[draw=none, fill=gray!50] (-.5,.5) --(-.5, -1.5) -- (-1, -1.5) -- (-1.25, .2) -- cycle;
\leapenv

\draw[fill] (-.5, 0)  circle (1.5pt);
\node[above right] at (-.5, 0)  {$\cop_1 = \sigma_0$};

\draw[fill] (-.85,-.4) circle (1.5pt);
\node[above] at (-.85,-.4) {$\cop_2$};

\draw[fill] (.41, -.3) circle (1.5pt);
\node[right] at (.41, -.3)  {$\rob$};

%\draw (-.5, -1.5) -- (-.5, .27);
\draw (-.5, -1.5) -- (.41, -.3);

%\draw[thick, fill=gray] (-.5, .51) circle (1.5pt);
%\node[above] at (-.5, .51) {$\sigma_1$};

%\draw (1,0) circle (.133in);
%\draw (-.5, -1.5) circle (.59in);

\node at (1,-2) {(b)};

\end{tikzpicture}

\\

%%% three
\begin{tikzpicture}

\draw[draw=none, fill=gray!50] (-.5,.5) --(-.5, -1.5) -- (-1, -1.5) -- (-1.25, .2) -- cycle;

\leapenv

\draw[fill] (-.5, .51)  circle (1.5pt);
\node[above left] at (-.5, .51)  {$\cop_1 = \sigma_0$};

\draw[fill] (.4,-.75) circle (1.5pt);
\node[above] at (.4,-.75)  {$\cop_2$};

\draw[fill] (1.28, -.5) circle (1.5pt);
\node[right] at (1.28, -.5)   {$\rob$};

\draw (.5,-1.5) --  (.78, -.8) -- (1.28, -.5);
\draw (-.5,-1.5) -- (1,-.25);

\draw[thick, fill=gray] (1,-.25) circle (1.5pt);
\node[left] at (1,-.25) {$\sigma_1$};

%\draw (.78, -.8) circle (.23in);
%\draw (-1, -1.5) circle (.95in);
%\draw (-.5, -1.5) circle (.8in);

\node at (1,-2) {(c)};

\end{tikzpicture}

&&

%%% four
\begin{tikzpicture}

\draw[draw=none, fill=gray!50] (-1.25, .2)  -- (-1, -1.5)  -- (.5, -1.5) -- (1,-.25) -- (1,0) -- (1.17,1.21); 

\leapenv

\draw[fill] (-.5, .26)  circle (1.5pt);
\node[right] at (-.5, .26)  {$\cop_1 = \sigma_0$};

\draw[fill] (1.1, -.75) circle (1.5pt);
\node[right] at (1.1, -.75)   {$\rob$};

\draw (.5,-1.5) --  (.78, -.8) -- (1.1, -.75);
\draw (-.5,-1.5) -- (.925, -.45);

\draw[thick,fill] (.90, -.48) circle (1.5pt);
\node[right] at (.90, -.48) {$\cop_2 = \sigma_1$};

%\draw (.78, -.8) circle (.13in);
%\draw (-.5, -1.5) circle (.69in);

\node at (1,-2) {(d)};

\end{tikzpicture}

\\

%%% five
\begin{tikzpicture}

\draw[draw=none, fill=gray!50] (1.25,1.25) -- (2, 1)  -- (2.25, .25) -- (1,0) -- (1,-.25) -- (2.5, -.35) -- (2.5,-1.5) -- (-1, -1.5) -- (-1.25, .2)  -- (1.25,1.25);
\leapenv

\draw[fill] (2.5, -.5)  circle (1.5pt);
\node[left] at (2.5, -.5)  {$\cop_1 = \sigma_2$};

\draw[fill] (3.35, -1.03) circle (1.5pt);
\node[right] at (3.35, -1.03)   {$\rob$};

\draw (2.5, -1.5) -- (3.25, -1.2) -- (3.35, -1.03);

\draw[thick, fill] (1.08, .7) circle (1.5pt);
\node[below left] at (1.08, .7) {$\cop_2 = \sigma_1$};

%\draw (.5, -1.5) circle (.9in);
%\draw (3.25, -1.2) circle (.2);
%\draw (2.5, -1.5) circle (1);
%\draw (1.25, -1.2)  circle (1.45);

\draw (.5, -1.5) -- (1.25, -1.2) -- (2.5, -.5);

\node at (1,-2) {(e)};

\end{tikzpicture}

&&

%%% six
\begin{tikzpicture}

\draw[draw=none, fill=gray!50] (1.25,1.25) -- (2, 1)  -- (2.25, .25) -- (1,0) -- (1,-.25) -- (2.5, -.35) -- (2.5,-1.5) -- (-1, -1.5) -- (-1.25, .2)  -- (1.25,1.25);
\leapenv

\draw[fill] (2.5, -.38)  circle (1.5pt);
\node[above] at (2.5, -.38)  {$\cop_1 = \sigma_2$};

\draw[fill] (3.5, -.8) circle (1.5pt);
\node[right] at (3.5, -.8)   {$\rob$};

\draw (2.5, -1.5) -- (3.25, -1.2) -- (3.5, -.8);

\draw[thick, fill] (2.9, -1.327) circle (1.5pt);
\node at (2.75, -1.13) {$\cop_2$};

%\draw (2.5, -1.5) circle (1);

%\draw (.5, -1.5) -- (2.5, -.5);

\node at (1,-2) {(f)};

\end{tikzpicture}

\end{tabular}
\end{center}

\caption{An example leapfrog pursuit in the environment shown in Figure \ref{fig:leap} using path projections $\sigma_0, \sigma_1, \sigma_2$ rooted at $u,v,w$, respectively.  (a) The game begins and $\cop_1$ moves until (b) $p_1$ clears the first simply connected region $Q_0$  using lion's strategy. (c) Next,  $\cop_2$ plays lion's strategy (rooted at $u$) in $Q_1 \backslash Q_0$ against the pursuer projection $\sigma_1$ until (d) $\cop_2$ clears  this region. This also prevents $\rob$ from moving between different components of $Q \backslash Q_1$. After that,  (e)  $\cop_1$ leapfrogs to clear  $Q_3 \backslash Q_2$and finally, (f) $\cop_2$ uses lion's strategy (rooted at $w$) in the last region $Q_4 \backslash Q_3$ to capture $e$. }

\label{fig:leapfrog-example}

\end{figure}
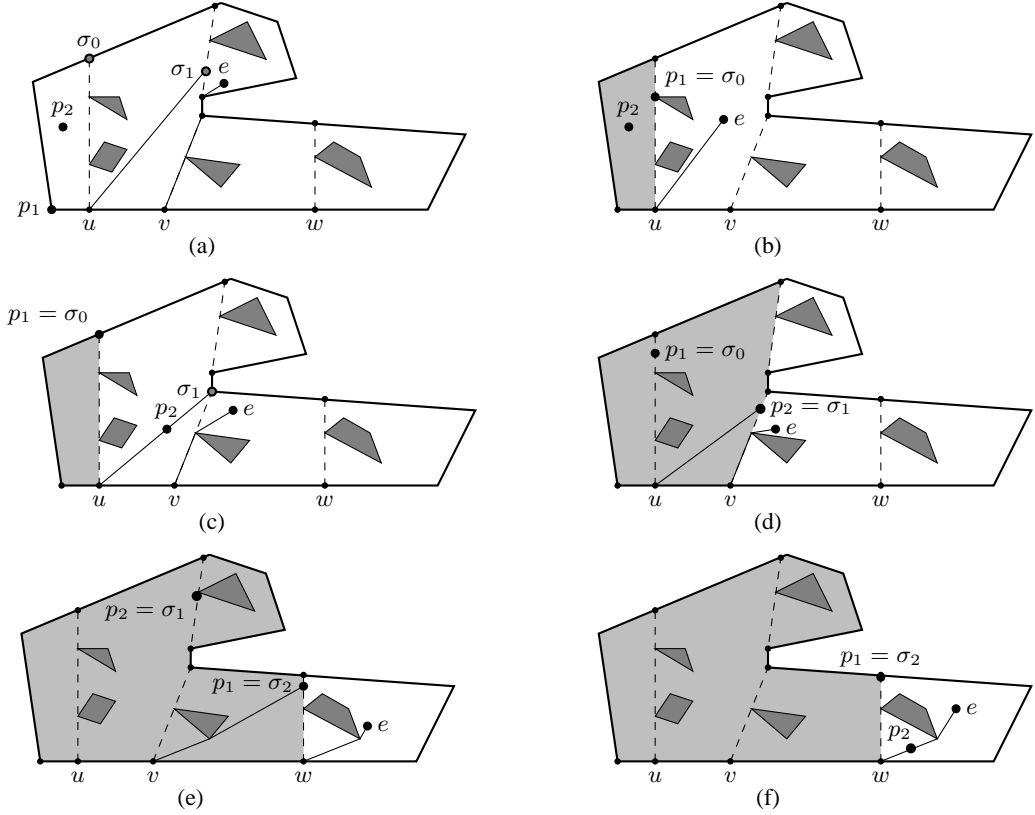

\begin{proofof}{Theorem \ref{thm:leapfrog}}
% New family of projections.
Note that the condition (L2) does not require that projections agree wherever the boundaries of the subregions intersect. The first order of business is to construct a new family of projections $\sigma_i: P \rightarrow Q_i$ for which this is the case.
We define these projections recursively, starting at $k-1$. We set $\sigma_{k-1} = \pi_{k-1}$ and then define
$\sigma_i = \pi_i \circ \sigma_{i+1}$ for $k-2 \geq i \geq 0$. The function  $\sigma_i$ is a projection: for $x,y \in P$, we have
\begin{eqnarray*}
d_{Q_0} (\sigma_0(x), \sigma_0(y) ) &=& d_{Q_0} ( \pi_0 (\sigma_1(x)), \pi_0(\sigma_1(y))) \\
& \leq & d_{Q_1} ( \sigma_1(x), \sigma_1(y) )
\, = \,
d_{Q_1} ( \pi_1(\sigma_2(x)), \pi_1(\sigma_2(y) )) \\
&\leq & \cdots  \, \leq \, d_{Q_{k-1}} (\pi_{k-1}(x), \pi_{k-1}(y)) \leq d_P(x,y).
\end{eqnarray*}
By construction, these recursively defined projections satisfy the conditions of Corollary \ref{cor: piecewise projection}. 
Most importantly for us, if $\sigma_{i+1}(x) \in \partial Q_i \cap \partial Q_{i+1}$ then $\sigma_i(x)=\sigma_{i+1}(x)$. In other words, suppose that $\cop_1$ is using  projection $\sigma_i$ and $\cop_2$ is using projection $\sigma_{i+1}$. If an evader projection is on the shared boundary $\partial Q_i \cap \partial Q_{i+1}$, then both pursuers agree on its location. This continuity is crucial to the leapfrog strategy, since it allows $\cop_1$ to react when $\cop_2$ sees the evader projection cross through $\partial Q_i$.

% The leapfrog strategy.
The leapfrog strategy proceeds as follows.
First, $\cop_1$ clears $Q_0$ with respect to the projection $\sigma_0: P \rightarrow Q_0$, as described in Lemma \ref{lemma:clear}.  Assume inductively that one pursuer, say $\cop_1$, currently guards $Q_{i}$. 
Next, $\cop_2$ works to clear $Q_{i+1}$ with respect to $\sigma_{i+1}: P \rightarrow Q_{i+1}$. 
When $Q_{i+1}\setminus Q_i$ is not connected, this clearing movement  requires $\cop_1$ as well. Indeed, if $\rob$ tries to move between regions then $\cop_1$ can respond with capture, see Figure \ref{fig:leapfrog-example} (d). The guarding $\cop_1$ captures $e$ whenever the evader moves between components of $Q_{i+1}\setminus Q_i$. Indeed,
if $Q_{i+1}\setminus Q_i$ is not connected, then we must have  $\partial Q_i \cap \partial Q_{i+1} \neq \emptyset$.  
 If a shortest path between $\sigma_{i+1}(\rob^{t-1})$ and $\sigma_{i+1}(\rob^{t})$ intersects $Q_i$ (and therefore intersects $\partial Q_i \cap \partial Q_{i+1}$), then $\cop_1$ can immediately respond by the capture move $\cop_1^{t+1} = \rob^{t}$ by Corollary \ref{cor:guard}. Indeed,  $\sigma_{i+1}(e)$ moves between components of $Q_{i+1}\setminus Q_i$ if and only if $e$ moves between components of $Q_{i+1} \backslash Q_i$. By the construction of the projections $\sigma_0, \ldots , \sigma_{k-1}$, the guarding $\cop_1$ is in position to capture the evader in response to this boundary crossing.
This means that the evader cannot move between the components of of $Q_{i+1} \backslash Q_i$ without being captured.

% Formal argument of why leapfrog works.
Let us return to $\cop_2$'s attempt to clear $Q_{i+1}$. 
%We may assume that $\pi_{i+1}(\rob) \notin \partial Q_i$ (otherwise $\cop_1$ can start guarding $Q_{i+1}$ in a single move as described above). 
Let $R_1, R_2, \ldots , R_k$ be the simply connected components of $Q_{i+1} \backslash Q_i$, and say that $\sigma_{i+1}(\rob) \in R_1$. 
While $\cop_1$ guards $Q_i$, pursuer $\cop_2$ moves into $R_1$ and tries to clear this region.
If the projection $\sigma_{i+1}(\rob)$ ever leaves $R_1$, then $\cop_1$ can immediately respond by capturing $e$ %moving onto $\pi_{i+1}(\rob)$ and thereby guard $Q_{i+1}$ 
(because the evader's projection moves through $\partial Q_i$, as described above). Otherwise, the projection $\sigma_{i+1}(\rob)$ always remains in $R_1$, so $\cop_2$ 
can capture this position  by Lemma \ref{lemma:clear}.
At this point, $\cop_2=\sigma_{i+1}(\rob)$. If $\rob = \sigma_{i+1}(\rob)$, then the evader is caught. Otherwise, $\cop_2 = \sigma_{i+1}(\rob) \in \partial Q_{i+1}$, and $\cop_2$ switches to  guarding $Q_{i+1}$. This releases $\cop_1$ to start  clearing $Q_{i+2}$. This leapfrogging continues until the evader is caught, which must occur when the pursuers finally control $Q_k \backslash Q_{k-1}$.
\end{proofof}

\begin{proofof}{Corollary \ref{cor:leapfrog}}
Since $Q_0$ is geodesically convex and $\partial Q_0$ contains two points $u,v \in \partial P$, the boundary $\partial Q_0$ is a minimal loop. By Lemma \ref{lemma:minloop}, there is a path projection $\pi_0: P \rightarrow Q_0$. Likewise,  $u,v \in \partial Q_i$ for every $0 \leq i \leq k$, so there is a path projection $\pi_i : P \rightarrow Q_i$. This is a family of projection functions required by  (L2) in Theorem \ref{thm:leapfrog}.
\end{proofof}

We conclude this section by giving an upper bound on the time to capture of the leapfrog strategy. The leapfrog strategy repeatedly uses lion's strategy in simply connected environments. Isler et. al. \cite{isler05tro} prove that in a simply connected polygon $R$, lion's strategy completes in time $O(m \cdot \mbox{diam}(R))$ where $m$ is the number of vertices of $R$ and $\mbox{diam}(R) = \max_{u,v} d_R(u,v)$. Therefore, lion's strategy completes in time $O( n \cdot \mbox{diam}(P))$
for each of $Q_0$ and $Q_{i+1} \backslash Q_i$, $0 \leq i \leq k-1$. Since $k \leq h+1$ (where $h$ is the number of holes), the leapfrog strategy completes   in time $O( n \cdot h \cdot \mbox{diam}(P))$

%%%%%%%%%%%%% Monotone 2-sweepable polygons %%%%%%%%%%%%
\section{Examples of Leapfrog Environments}
\label{sec:sweep}

\def\trienv{

% outer
\draw (0,0) -- (4.9,0) -- (2.45,3.5) -- cycle;  

% left
\begin{scope}[shift={(.65, .35)}]

\draw[fill=gray] (0,0) -- (1.4,0) -- (0.7,1) -- cycle;  

\end{scope}

% above
\begin{scope}[shift={(1.74,1.85)}]

\draw[fill=gray] (0,0) -- (1.4,0) -- (0.7,1) -- cycle;  

\end{scope}

% right
\begin{scope}[shift={(2.85,.35)}]

\draw[fill=gray] (0,0) -- (1.4,0) -- (0.7,1) -- cycle;  

\end{scope}

% middle
\draw[fill=gray] (1.54, 1.6) -- (3.36, 1.6) -- (2.44, .35);

}

Theorem \ref{thm:leapfrog} and Corollary \ref{cor:leapfrog} describe the characteristics of a leapfrog decomposition  $Q_0 \subset Q_1 \subset \cdots \subset Q_k = P$. Identifying leapfrog environments (or more generally, two-pursuer-win environments) remains an open problem. In this section, we consider two fundamental subfamilies of leapfrog environments and give polynomial time algorithms that verify membership. % in these families.

\subsection{Minimal path leapfrog environments}

Our first family of leapfrog environments arises when  the path construction methods of three-pursuer Minimal Path Strategy described in Bhadauria et al.~\cite{bhadauria+klein+isler+suri} actually produces a leapfrog decomposition. We refer the reader to Section 4 of that paper for proofs and discussion of the more subtle points. 
Figure \ref{fig:trienv} shows an example of a successful minimal path leapfrog decomposition resulting from their minimal path constructions.
Start with anchor points $u, v$ which are vertices on the outer boundary $\partial P$.
Choose path $\Pi_0$ to be a shortest $(u,v)$-path in $P$. 
Given $\Pi_0$, we next find a $(u,v)$-shortest path $\Pi_1$ under the restriction that $\Pi_1$ includes at least one vertex of $P$ that is not in $\Pi_0=Q_0$. 
%In terms of the visibility graph, we want a path $\Pi_1$ that differs in at least one visibility edge from $\Pi_0$.
 Furthermore, we restrict ourselves paths that do not  loop around obstacles in $P$. Let $\partial Q_1 = \Pi_0 \cup \Pi_1$ with $Q_1$ its interior. If $Q_1$  is simply connected then we continue; otherwise our attempt to build a leapfrog decomposition fails for the choice of $(u,v)$.

Next, suppose that we have already successfully defined $Q_0 \subset Q_1 \subset \cdots \subset Q_{i-1}$.
% so that two cops can clear and guard $Q_{i-1}$. 
Set  $R := P \backslash \mathrm{int}(Q_{i-1})$ and note that if a pursuer were guarding $Q_{i-1}$ then  $P_e \subset R$.  Find a $(u,v)$-shortest path $\Pi_i \in R$ that includes at least one vertex of $P$ that is not in $Q_{i-1}$.   Let $\partial Q_i$ be the outermost boundary of $Q_{i-1} \cup \Pi_i$ and let $Q_i \subset P$ be its closure in $P$. If $Q_i \backslash Q_{i-1}$ is simply connected then we continue; otherwise our attempt fails for anchors $(u,v)$. It is easy to see that the path $\Pi_i$ is a minimal path in $P_e$.
This process creates a leapfrog decomposition if  we can continue the construction until $Q_k = P$.

% Sweepable polygons.
\begin{mydef}[Minimal Path Leapfrog Environment]
	A polygonal environment $P$ is a \emph{minimal path leapfrog environment} if the series of paths $\Pi_0, \Pi_1, \ldots , \Pi_k$ produced by the Minimal Path Strategy induces a leapfrog decomposition $Q_0 \subset Q_1 \subset \cdots \subset Q_k$.
\end{mydef}

Minimal path leapfrog environments can be identified in polynomial time.
The visibility graph $G(P)$ of the polygon (whose nodes correspond to the $n$ vertices of $P$) captures the paths described above, and this graph can be constructed in $O(n^2 \log n)$ time (we only need to construct this graph once). Finding the loop-free paths described above requires finding shortest paths and second-shortest paths in $G(P)$ or a subgraph of $G(P)$. Each such path can be found using Yen's algorithm \cite{yen} for finding short loop-free paths in graphs. This algorithm runs in $O(n m + n \log n)$ time  where $m$ is the number of edges in $G(P)$.  We must iterate at most $n$ times (since each path $\Pi_i$ visits at least one previously unvisited vertex). We might need to try all ${n \choose 2}$ possible choices for anchor points, so the worst case time bound would be $O(n^4 m + n^4 \log n) = O(n^6)$. 
%Of course, even  if this process fails for all choices of $u,v$, there may still be a leapfrog decomposition that can be constructed another way.

\begin{figure}[ht]

%%%%%%
%%%%%%

\begin{center}

\begin{tikzpicture}[scale=.65]

\node at (.5,2) {$P$};

\trienv

%%% 1
\begin{scope}[shift={(6,0)}]

\node at (.5,2) {$Q_0$};
\trienv

\draw[very thick] (0,0) -- (2.45,3.5);

\draw[fill] (0,0) circle (3pt);
\draw[fill] (2.45,3.5) circle (3pt);

\node[below] at (0,0) {$u$};
\node[right] at (2.45,3.5) {$v$};

\end{scope}

%%% 2

\begin{scope}[shift={(12,0)}]

\node at (.5,2) {$Q_1$};
\draw[very thick, fill=gray!50] (0,0) -- (2.45,3.5) -- (1.54, 1.6) -- cycle;

\trienv

\draw[fill] (1.54, 1.6) circle (3pt);

\end{scope}

%%%% 3
\begin{scope}[shift={(18,0)}]

\node at (.5,2) {$Q_2$};
\draw[very thick, fill=gray!50] (0,0) -- (2.45,3.5) -- (1.54, 1.6) -- (1.35, 1.35) -- cycle;

\trienv

\draw[fill] (1.35, 1.35)  circle (3pt);

\end{scope}

%%%% 4
\begin{scope}[shift={(0,-5)}]

\node at (.5,2) {$Q_3$};

\draw[very thick, fill=gray!50] (0,0) -- (2.45,3.5) -- (1.74,1.85) -- (1.35, 1.35) -- cycle;

\trienv

\draw[fill] (1.74,1.85)   circle (3pt);

\end{scope}

%%%% 5
\begin{scope}[shift={(6,-5)}]

\node at (.5,2) {$Q_4$};
\draw[very thick, fill=gray!50] (0,0) -- (2.45,3.5) -- (2.44, 2.85) -- (1.35, 1.35) -- cycle;

\trienv

\draw[fill] (2.44, 2.85)  circle (3pt);

\end{scope}

%%%%% 6
\begin{scope}[shift={(12,-5)}]

\node at (.5,2) {$Q_5$};

\draw[very thick, fill=gray!50] (0,0) -- (2.45,3.5) -- (2.44, 2.85)  -- (1.35, 1.35) -- (.65, .35) -- cycle;

\trienv

\draw[fill] (.65, .35)  circle (3pt);

\end{scope}

%%%%% 7
\begin{scope}[shift={(18,-5)}]

\node at (.5,2) {$Q_6$};

\draw[very thick, fill=gray!50] (0,0) -- (2.45,3.5) -- (3.36, 1.6) -- (2.44, .35)   -- cycle;

\trienv

\draw[fill] (3.36, 1.6) circle (3pt);
\draw[fill] (2.44, .35)  circle (3pt);

\end{scope}

%%%% 8 
\begin{scope}[shift={(0,-10)}]

\node at (.5,2) {$Q_7$};

\draw[very thick, fill=gray!50] (0,0) -- (2.45,3.5) -- (3.55,1.35) --   (2.44, .35)  -- cycle;

\trienv

\draw[fill] (3.55,1.35)   circle (3pt);

\end{scope}

%%%% 9
\begin{scope}[shift={(6,-10)}]

\node at (.5,2) {$Q_8$};

\draw[very thick, fill=gray!50] (0,0) -- (2.45,3.5) -- (3.55,1.35) -- (2.85, .35)   -- cycle;

\trienv

\draw[fill] (2.85, .35)   circle (3pt);

\end{scope}

%%% 10
\begin{scope}[shift={(12,-10)}]

\node at (.5,2) {$Q_9$};

\draw[very thick, fill=gray!50] (0,0) -- (2.45,3.5)  -- (4.25, .35)  -- cycle;

\trienv

\draw[fill] (4.25, .35)   circle (3pt);

\end{scope}

%%% 11
\begin{scope}[shift={(18,-10)}]

\node at (.5,2) {$Q_{10}$};

\draw[very thick, fill=gray!50] (0,0) -- (2.45,3.5)  -- (4.9, 0)  -- cycle;

\trienv

\draw[fill] (4.9, 0)   circle (3pt);

\end{scope}

\end{tikzpicture}

\end{center}

\caption{A minimal path leapfrog environment. The leapfrog decomposition $Q_0 \subset Q_1 \subset \cdots \subset Q_{10}$ where $Q_i$ is created by finding a shortest $(u,v)$-path in $P \backslash \mathrm{int}(Q_{i-1})$ for which $Q_i \backslash Q_{i-1}$ is simply connected. Each successive figure shows the  vertices of $P$ appearing in a $Q_i$ for the first time.}

\label{fig:trienv}

\end{figure}
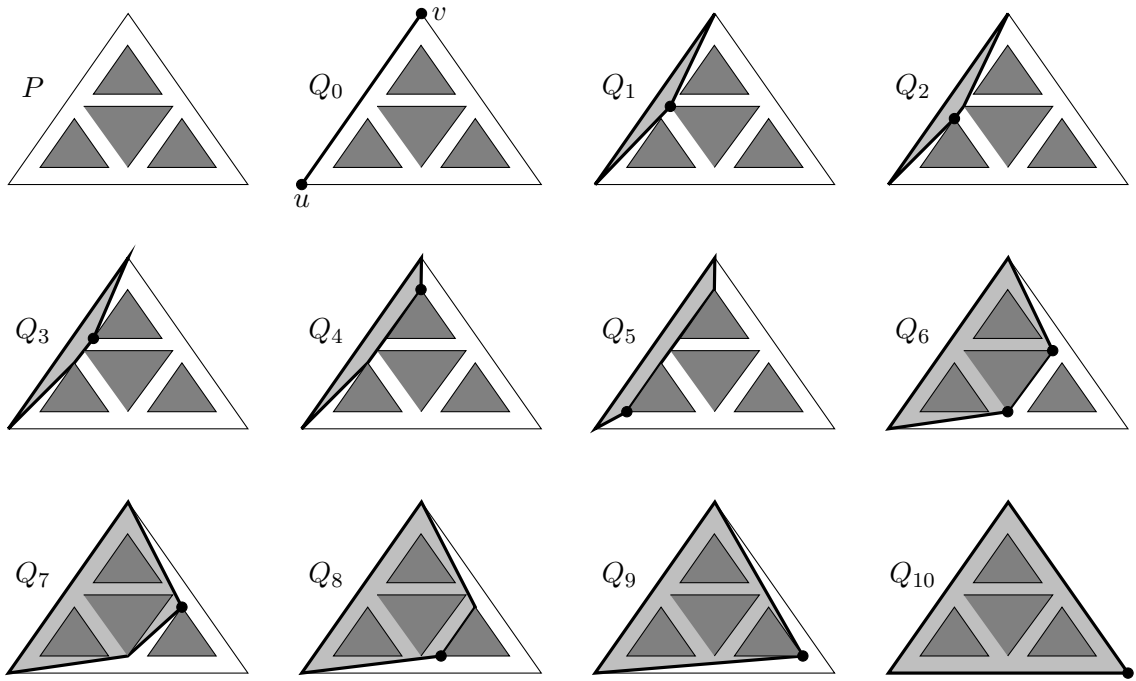

\subsection{Strictly two-sweepable environments}

%-----------------------------------------------------------------------
% Intro
%-----------------------------------------------------------------------

% What we aim to do.
In this section, we describe an even simpler class of environments for which  the the two-pursuer leapfrog strategy is successful.  Specifically, we show that a \emph{strictly two-sweepable environment} $P$  has a family of nested sub-environments  $Q_0 \subset Q_1 \subset \cdots Q_k = P$ that satisfy the hypothesis of Corollary~\ref{cor:leapfrog}. Finally, we give an $O(n^2)$  algorithm that determines  if a given polygon $P$ is strictly two-sweepable, based on finding a specific path in the dual of $P$. This algorithm explicitly constructs the family 
$Q_0, Q_1, \dots, Q_k$ for use in the leapfrog strategy.

%-----------------------------------------------------------------------
% Sweepable polys: defs and 2-cop winnable.
%-----------------------------------------------------------------------
%\subsection{Sweepable Polygons}

In \cite{bose}, Bose and Kreveld present a method for determining if a polygon can be monotonically swept by a line. We extend these ideas to a polygonal environment with obstacles.  Intuitively speaking, the obstacles of these polygons are well-separated, so we can find our leapfrog partition using straight-line boundaries.  We begin with two definitions. 

% Sweepable polygons.
\begin{mydef}[Sweepable Polygons \cite{bose}]
	A polygon $P$ is \emph{sweepable} if a line $\ell$ can be swept continuously over $P$ such that each
	intersection of the line and $P$ is a convex set. We call such a line a \emph{sweep line} of $P$.
	Polygon $P$ is  \emph{strictly sweepable} if there exists a sweep line such that no portion of $P$ is swept over more than once.
\end{mydef}

% strictly two-sweepable environments.
\begin{mydef}[Strictly Two-Sweepable Environments]
	A polygonal environment $P$ is said to be \emph{two-sweepable} if a line can be swept continuously over
	$P$ such that each cross section of $P$ with respect to this line is the disjoint union of at most two convex sets.		
	 Environment $P$ is \emph{strictly two-sweepable} if $P$ is two-sweepable and
	its boundary polygon $B_P$ is strictly sweepable.	
	Equivalently, $P$ is strictly two-sweepable if  there exists a sweep line $\ell$ such that
	no portion of $P$ is swept more than once and the intersection of  $\partial P$ and each cross section of  $P$, with respect to $\ell$,
	contains at most two points.
\end{mydef}			

%{\bf Remark: } this def is exactly that for a ``strictly simply two-sweepable\rq\rq{} polygon (in our earlier terminology)
%although in a much more compact form.

Note that the environment in Figure \ref{fig:trienv} is not strictly two-sweepable, so this family is distinct from minimal-path leapfrog environments. This brings us to the main result of this section.
			
% M2S = two-cop win.		
\begin{theorem}
	\label{thm:2sweep}
	If  the polygonal environment $P$ is strictly two-sweepable then $P$ is two-pursuer-win.
\end{theorem}

\begin{proof}
We describe how to construct a family of nested sub-environments satisfying the hypothesis of
Corollary~\ref{cor:leapfrog} using the movement of the sweep line.
% Assumptions and sweep-partition.
Let $P$ be a strictly two-sweepable polygonal environment with $n$ vertices $\{v_1, \dots, v_n\}$ and sweep line $\ell$.
% Sweep lines
Let $\ell_1, \dots, \ell_n$ denote the positions of the sweep line $\ell$ intersecting each vertex of $P$.
Here the vertex labels (and cross section labels) are ordered non-increasingly with respect to the order in which they are
swept by $\ell$; the fact that $P$ is strictly two-sweepable ensures that such an ordering is well-defined.

%%% Sweep partition example image.
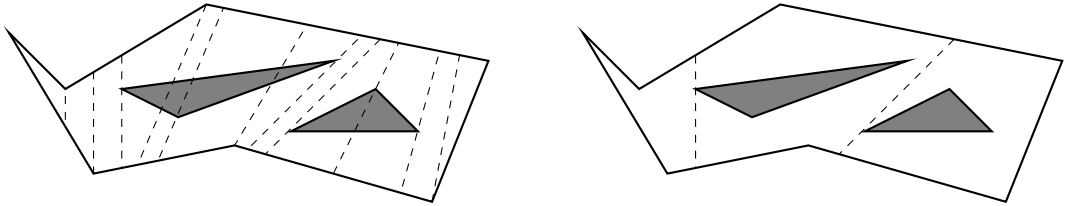
\begin{figure}[ht]
\begin{center}
\begin{tabular}{ccc}
\begin{tikzpicture}[scale=.75]

\draw[thick] (.5,1) -- (1.5,0) -- (4,1.5) -- (9,.5) -- (8, -2) -- (4.5, -1) -- (2,-1.5) -- cycle;
\draw[thick, fill=gray] (2.5,0) -- (6.25,.5) -- (3.5,-.5)- - cycle; 
\draw[thick, fill=gray] (5.5,-.75) -- (7,0) -- (7.75,-.75)- - cycle; 

\begin{scope}
  \clip  (.5,1) -- (1.5,0) -- (4,1.5) -- (9,.5) -- (8, -2) -- (4.5, -1) -- (2,-1.5) -- cycle;
  
  \draw[dashed] (1.5, 1) -- (1.5,-2);
  \draw[dashed] (2, 2) -- (2,-2);
  \draw[dashed] (2.5, 2) -- (2.5,-2);
  \draw[dashed] (4, 1.5) -- (2.5,-2);
  \draw[dashed] (4.33, 1.5) -- (2.83,-2);
  \draw[dashed] (6, 1.5) -- (4.5,-1);
  
  \draw[dashed] (3.8,-2) -- (7.8,2);
  \draw[dashed] (4.2,-2) -- (8.2,2);  
  
  \draw[dashed] (6,-2) -- (8,2);    
  \draw[dashed] (7.4,-2) -- (8.5,2);    
  \draw[dashed] (8,-2) -- (8.75,2);      
  
\end{scope}

\end{tikzpicture}
&
\hspace{.15in}
&
\begin{tikzpicture}[scale=.75]

\draw[thick] (.5,1) -- (1.5,0) -- (4,1.5) -- (9,.5) -- (8, -2) -- (4.5, -1) -- (2,-1.5) -- cycle;
\draw[thick, fill=gray] (2.5,0) -- (6.25,.5) -- (3.5,-.5)- - cycle; 
\draw[thick, fill=gray] (5.5,-.75) -- (7,0) -- (7.75,-.75)- - cycle; 

\begin{scope}
  \clip  (.5,1) -- (1.5,0) -- (4,1.5) -- (9,.5) -- (8, -2) -- (4.5, -1) -- (2,-1.5) -- cycle;

  \draw[dashed] (2.5, 2) -- (2.5,-2);
 
  \draw[dashed] (4.2,-2) -- (8.2,2);

\end{scope}

\end{tikzpicture}
\end{tabular}
\end{center}
\caption{\small{A region partitioned into convex regions by a sweeping line.  Left: the sweep lines $\ell_1, \ldots \ell_n$.  Right: the partition $\cP$.}}
\label{fig:twosweeppartition}
\end{figure}
%%% Sweep partition example image.

% Partition to nested subpolys.
We  construct our family of nested subregions.
Let $\cH = \{H_1, H_2, \dots, H_h\}$ denote the set of obstacles in the environment $P$;
as before, the obstacles are labeled with respect to the order in which are swept by $\ell$.
% New sweep partition.
Let $ \{ \lambda_{1},  \ldots, \lambda_{h} \} \subset \{ \ell_1, \ldots , \ell_n \}$ denote the positions of the sweep line $\ell$
first intersecting each obstacle $H_1, H_2, \dots, H_h$ as $\ell$  sweeps  $P$.
We also define  $\lambda_{0} = \ell_1,$ $ \lambda_{{h+1}} = \ell_{n}$, corresponding to the first and last boundary vertices encountered by $\ell$ during the sweep.
Let $\cP = \{ P_0, P_1, \dots,  P_h \}$ denote the sub-environments of $P$, where $P_k$ is inscribed by $\lambda_{k}$ and $\lambda_{k+1}$.
The family $\cP$ is a partition  of $P$ where
each sub-environment $P_i$ is simply connected, as shown
 in Figure \ref{fig:twosweeppartition}.
Finally, construct our family $\cQ$ of nested sub-environments by taking 
inductively by taking $Q_0 = P_0$ and $Q_i = P_i \cup Q_{i-1}$ for $1 \leq i \leq h$.
% Nested subenvs are leapfroggable.
It is easy to check that the family $\cQ$
satisfies the hypothesis of Corollary~\ref{cor:leapfrog}.
\end{proof}

%---------------------------------------------------------------------------------------------------------------------
% Duality.
\subsection{Duality algorithms for finding sweepable polygons}
%---------------------------------------------------------------------------------------------------------------------

% Preamble.
We conclude  with a practical discussion of how to determine if a given
polygonal environment is strictly two-sweepable (and therefore two-pursuer win).
The technique  extends the duality-based method of Bose and Kreveld in \cite{bose} 
from simply connected polygons to  
 polygonal environments containing obstacles.
% Dual transformation.
We summarize Bose and Kreveld\rq{}s results here and direct the reader to \cite[Section 4]{bose} for a
more thorough treatment of the topic.

Let $P$ be a simply connected polygon.  To determine if $P$ is
sweepable and/or strictly sweepable, Bose and Kreveld  consider the movement of a proposed sweep line
in the dual of $P$. 
% Dual mapping for vertices.
We consider the duality transform that maps a given point $p = (a,b)$ to the line $D_p = \{(x,y): ax - b\}$ with slope $a$ and $y$-intercept $-b$,
and maps each line $L = \{(x,y): y = mx + c\}$ to the point $(m, -c)$.
Each edge of $P$ is mapped to the face of the dual inscribed
by the two lines corresponding to its endpoints; we call the pair of faces of the dual corresponding to an edge in the primal a {\it double wedge}.
We call the lower and upper envelopes of all faces in the dual the {\it start face} and {\it end face} respectively.

% Sweep lines in P translate to paths in the dual.
The movement of a line swept continuously across $P$ dualizes to a path from the start face to the end face in the dual
arrangement. Each cross section of $P$ defined by the line $\ell$ is mapped to a point in the dual in the intersection
of all double wedges corresponding to the edges intersected by $\ell$.
% vertical lines.
Rotation of a line past a vertical position is represented by the path \lq{}\lq{}jumping\rq\rq{} from one unbounded face
of the dual arrangement to the opposite unbounded face in the double wedge corresponding to the
edges of $P$ intersected by the vertical line.
Each such jump corresponds to a change in orientation of the line as it sweeps $P$.
Crucially, as the line sweeps $P$ its trajectory in the dual is restricted to an even number of
such jumps to ensure that the orientation of the line is maintained.
% Double wedges.

Recall that $P$ is sweepable if and only if a line can be continuously across $P$ such that each cross section
of $P$ is a convex set. Equivalently, $P$ is sweepable if and only if there exists sweep line intersecting at most two edges of $P$ at a time.
This sweep line corresponds to a path in the dual arrangement that does not traverse the intersection of more than two double wedges.
We call a face in the dual arrangement in which at least three double wedges overlap a {\it forbidden face}.
Therefore, $P$ is sweepable if and only if there is a path from
the start face to the end face in the dual arrangement avoiding all forbidden faces.
We call such a path in the dual arrangement a {\it sweep path}. 
Both the dual arrangement of $P$ and the set of all forbidden faces can be constructed in $O(n^2)$ time
(cf. \cite[Chapter 8]{compgeo}).
Performing depth-first search on the faces of the dual arrangement  also takes 
$O(n^2)$ time; indeed, the graph induced by the dual arrangement is planar and contains $O(n^2)$ nodes and arcs. 

%%%NEW POLYGON - SHOWING VERTICAL LINES/UNBOUNDED FACES
\begin{figure}[ht]
\begin{center}
\begin{tikzpicture}[scale=.4]

%------------------------------------------------------------------------------------------------
%%P
%------------------------------------------------------------------------------------------------
\begin{scope}[scale=1.5, domain=-3.5:3.5]
% Frame.
%\clip[draw] (-2.5, -1.5) rectangle (1.5, 1.5); 
\draw[very thin,color=gray] (-3,-2) grid (2, 2);
\draw[<->, dashed](0,-2.5) -- (0,2.5); 
\draw[<->, dashed](-3.5,0) -- (2.5,0); 

% Draw outline of P.
\draw[color=blue,very thick] (0,0) -- (0,1) -- (1,1) -- (1,-1)--(-2, -1) -- (-2,1) -- (-1, 1) -- (-1,0) --  cycle;

%% Vertices.
\draw[fill=red] (0,0)  circle(3pt);
\draw[fill=blue](0,1) circle(3pt);
\draw[fill=green] (1,1) circle(3pt);
\draw[fill=orange] (1, -1) circle(3pt);
\draw[fill=purple] (-2,-1) circle(3pt);
\draw[fill=yellow] (-2,1) circle(3pt);
\draw[fill= cyan] (-1, 1) circle(3pt);
\draw[fill=violet] (-1,0) circle(3pt);
\end{scope}

%------------------------------------------------------------------------------------------------
%% Dual of P
%------------------------------------------------------------------------------------------------
\begin{scope}[shift={(10,0)}, scale =1, domain=-5.2:5.2]
\pgfsetfillopacity{0.2}
%\% Grid.
\draw[very thin,color=gray] (-4,-4) grid (4,4); 
% Axes.
\draw[<->,dashed] (0,-4.5) -- (0,4.5); 
\draw[<->,dashed] (-4.5,0) -- (4.5,0); 

%% Duals of edges.
\begin{scope}
 \clip (-4.2,-4.2) rectangle (4.2,4.2);
\pgfsetfillopacity{0.25}

% (0,1)--(1,1)
\draw[fill=blue, ultra thin] (4,-1)--(-4,-1)--(-4,-4) -- (-3,-4) -- (4,3) -- cycle;
%%(1,1)--(1,-1)
\draw[fill=blue, ultra thin]   (3,4) -- (4,4) -- (4,3) --(-3,-4) -- (-4,-4) -- (-4,-3) -- cycle;
% (1,-1) -- ( -2, -1)
\draw[fill=blue, ultra thin] (-4,-3) -- (-4,4) -- (-1.5, 4) -- (2.5,-4) -- (4,-4) -- (4,4) -- (3,4) -- cycle;
% (-2,-1) -- (-2,1)
\draw[fill=blue, ultra thin] (-1.5,4) -- (2.5,-4) -- (1.5,-4) -- (-2.5,4) -- cycle;
%% (-2,1) -- (-1,1)
\draw[fill=blue, ultra thin] (-2.5,4) -- (1.5,-4) -- (3,-4) -- (-4,3) -- (-4,4) -- cycle;
%%% (-1,1) -- (-1,0)
\draw[fill=blue, ultra thin] (-4,3) -- (-4,4) -- (4,-4) -- (3,-4) -- cycle;
%%% (-1,0) -- (0,0)
\draw[fill=blue, ultra thin] (4,-4) -- (4,0) -- (-4,0) -- (-4, 4) -- cycle;
%% (0,0) -- (0,1)
\draw[fill=blue, ultra thin] (4,0) -- (4,-1) -- (-4,-1) -- (-4,0) -- cycle;
\end{scope}

% Duals of  vertices
\clip (-4.2,-4.2) rectangle (4.2, 4.2);

\draw[color=red, thick]  (-4.2,0) -- (4.2,0); %(0,0)
\draw[color=blue, thick]  (-4.2,-1) -- (4.2,-1); %(0,1)
\draw[color=green, thick]  plot(\x, {1*\x -1});  %(1,1)
\draw[color=orange, thick] plot(\x, {1*\x +1}); %(1,-1)
\draw[color=purple, thick] plot(\x, {-2*\x + 1}); %(-2, -1)
\draw[color=yellow, thick] plot(\x, {-2*\x -1}); %(-2,+1)
\draw[color= cyan, thick] plot(\x, {-1*\x -1}); %  (-1,1) 
\draw[color=violet, thick] plot(\x, {-1*\x}); % (-1,0)
\end{scope}

%------------------------------------------------------------------------------------------------
%% Forbidden faces and valid path 
%------------------------------------------------------------------------------------------------
\begin{scope}[shift={(21,0)}, scale =1, domain=-4.5:4.5]
\footnotesize
%\% Grid.
\draw[very thin,color=gray] (-4,-4) grid (4,4); 
% Axes.
\draw[dashed] (-4.2,0) -- (4.2,0) node[right] {}; 
\draw[dashed] (0,-4.2) -- (0,4.2) node[above] {}; 
%% Duals of edges.
\begin{scope}
 \clip (-4.2,-4.2) rectangle (4.2, 4.2);
\pgfsetfillopacity{0.5}

%% (0,1)--(1,1)
\draw[fill=blue, ultra thin] (-4,4) -- (0,0) -- (4,0) -- (4,-1) -- (0,-1) --  (-4,3) -- cycle;

% Duals of  vertices
\draw[color=red, thick]  (-4.2,0) -- (4.2,0); %(0,0)
\draw[color=blue, thick]  (-4.2,-1) -- (4.2,-1); %(0,1)
\draw[color=green, thick]  plot(\x, {1*\x -1});  %(1,1)
\draw[color=orange, thick] plot(\x, {1*\x +1}); %(1,-1)
\draw[color=purple, thick] plot(\x, {-2*\x + 1}); %(-2, -1)
\draw[color=yellow, thick] plot(\x, {-2*\x -1}); %(-2,+1)
\draw[color=cyan, thick] plot(\x, {-1*\x -1}); %  (-1,1) 
\draw[color=violet, thick] plot(\x, {-1*\x}); % (-1,0)

\end{scope}

%% Sweep path
\clip (-5,-5) rectangle (5, 5);
\node at (0,-3.5)  {start face} ;
\node at  (0, 3.5) {end face}; 
\draw[|->, very thick, densely dotted] (0.15, -3.2) to [out=75, in =180] (4, -1.5) node[right] {1};
\draw[->, very thick, densely dotted] (-4,-1.5) node[left] {1} to [out=0, in =360]  (-4, 1.5) node[left] {2} ;
\draw[->, very thick, densely dotted] (4, 1.5) node[right] {2} to [out= 180, in=-25]  (.5,3.2);

\end{scope}
\end{tikzpicture}
\end{center}
\caption{\small{Finding a path using unbounded faces in the dual. All vertices of the polygon are crossed. Since this jumping movement happens an even number of times (labeled 1 and 2), the sweeping line has the same left-right orientation in the initial position as in the ending position.}}
\label{fig:verticaldual}
\end{figure}
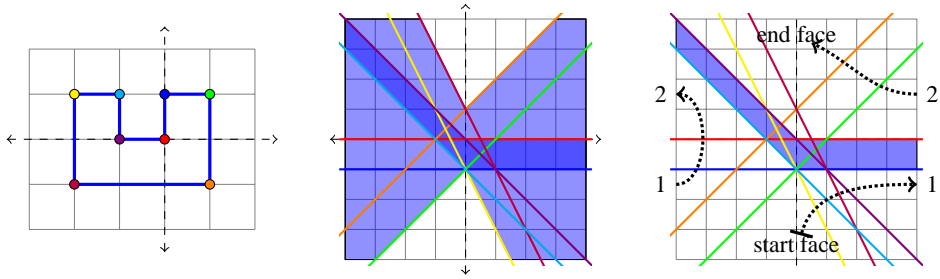

%END using dual faces diagram.

% Strictly sweepable polys.
A similar process can be used to determine if $P$ is strictly sweepable. We must take care when sweeping \emph{reflex vertices}, where the interior angle is greater than $\pi$ radians.  
In this case, we add additional vertices to $\partial P$ by temporarily extending the edges at reflex vertices into lines, and then adding a new vertex at each intersection of these lines with the boundary. We denote the resulting \emph{extended polygon} as $P\rq{}$.
The following lemma from \cite{bose} characterizes when a polygon $P$ is strictly sweepable
using the extension $P\rq{}$.

\begin{lemma}[{\cite[Lemma~4]{bose}}]
\label{lem: BK 4}
	Let $P$ be a simple polygon, and let $P\rq{}$ be its extended polygon.  $P$ is strictly sweepable if and only if $P\rq{}$ admits a sweep line that traverses each vertex exactly once.
\end{lemma}	

Algorithm \ref{alg:ssweepfind} determines whether a polygon is strictly sweepable and identifies the corresponding sweep line (if one exists). Extending $P$ to $P\rq{}$ can be performed in $O(n^2)$ time using a brute-force algorithm. The extended polygon  $P\rq{}$ contains $m = O(n)$ vertices, so
constructing the dual arrangement and performing depth-first search on the faces of the dual arrangement can be performed in $O(m^2) = O(n^2)$ time.

\begin{algorithm}
\caption{Strictly Sweepable Path Search \cite{bose}}
\label{alg:ssweepfind}
\begin{algorithmic}
\item[] Given polygon $P$.
\item[] Extend all reflex vertices to obtain $P\rq{}$ with $m$ vertices
\item[] Compute dual arrangement $D_{P\rq{}}$ and identify all forbidden faces of $D_{P\rq{}}$.
\item[] Apply depth first search on the faces of $D_{P\rq{}}$ to find the sweep path in $D_{P\rq{}}$
crossing the fewest lines.
\item[] If this path crosses exactly $m$ lines then $P$ is strictly
sweepable.
\end{algorithmic}
\end{algorithm}

%---------------------------------------------------------------------------------------------------------------------
% Duality.
%\subsection{Duality and an algorithm for identifying two-sweepable environments}
%---------------------------------------------------------------------------------------------------------------------

%%%SHOWING SIMPLY TWO SWEEP - ALL VERTICES OF P WITH FORBIDDEN(4) FACES, then map showing UNION of faces
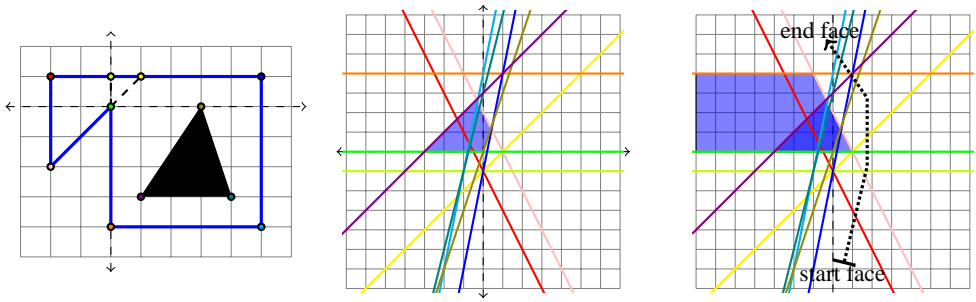
\begin{figure}[ht]
\footnotesize
\begin{center}
\begin{tikzpicture}[scale=.4]
%------------------------------------------------------------------------------------------------
%%P
%------------------------------------------------------------------------------------------------
\begin{scope}[shift = {(0, 1.5)}, domain=-6.5:6.5]
% Frame.
%\clip[draw] (-2.5,-4.5) rectangle (5.5,1.5); 
\draw[very thin,color=gray] (-3,-5) grid (6,2);
\draw[<->, dashed](0,-5.5) -- (0, 2.5); 
\draw[<->, dashed](-3.5,0) -- (6.5,0); 

% Draw outline of P.
\draw[color=blue,very thick] (-2, -2) -- (-2, 1) -- (5, 1) -- (5, -4) -- (0, -4) -- (0,0) -- (-2, -2);
%

% Extended vertices.
\draw[dashed, thick] (0,0) -- (1, 1);
\draw[fill = yellow, thick](1,1) circle(3pt);

\draw[dashed, thick] (0,0) -- (0,1);
\draw[fill=lime, thick] (0,1) circle(3pt);

% Vertices.
\draw[fill=red, thick](-2,1) circle(3pt);
\draw[fill=pink, thick](-2,-2) circle(3pt);
\draw[fill=green, thick](0,0) circle(3pt);
\draw[fill=orange, thick](0, -4) circle(3pt);
\draw[fill=cyan, thick](5, -4) circle(3pt);
\draw[fill=blue,  thick](5,1) circle(3pt);

% Obstacle.
% Fill.
\draw[fill=black] (1,-3) -- (3, 0) -- (4, -3);
% vertices.
\draw[fill=olive, thick] (3,0) circle (3pt);
\draw[fill=teal, thick](4,-3) circle (3pt);
\draw[fill=violet, thick](1,-3) circle (3pt);

\end{scope}

%------------------------------------------------------------------------------------------------
%% Dual of P with M2S forbidden faces.
%------------------------------------------------------------------------------------------------
\begin{scope}[shift={(12.375,0)},scale=.65,domain=-7.5:7.5]
%\begin{scope}[,scale=2.75,domain=-7.2:7.2]

% Grid.
\draw[very thin,color=gray] (-7,-7) grid (7,7); 
\draw[<->, dashed] (-7.5,0) -- (7.5,0) node[right] {}; 
\draw[<->, dashed] (0,-7.5) -- (0,7.5) node[above] {}; 

%% Dualization of P.
% Fill in forbidden faces.
\begin{scope}
\clip (-7.2,-7.2) rectangle (7.2,7.2);
\pgfsetfillopacity{0.5}
%%%%drawing/finding forbidden faces here
\draw[fill=blue, ultra thin] (-3,0) -- (-1/3, 8/3) -- (2/5, 6/5) -- (0,0) -- (-3,0) -- cycle;
\end{scope}

\begin{scope}
% Boundary and extended vertices
\clip (-7.2,-7.2) rectangle (7.2,7.2);
\draw[color=red, thick] plot (\x, {-2*(\x) -1});  %(-2,1)
\draw[color=pink, thick] plot (\x, {-2*(\x) +2});  %(-2,-2);
\draw[color=lime, thick] (-7.2, -1) -- (7.2, -1) ; %(0,1);
\draw[color=yellow, thick] plot (\x, {(\x) - 1}); %(1,1)
\draw[color=green, thick] (-7.2,0) -- (7.2, 0) ; %(0,0)
\draw[color=orange, thick] (-7.2, 4) -- (7.2, 4); %(0, -4).
\draw[color=cyan, thick] plot (\x, {5*(\x)+4}); %(5, -4).
\draw[color=blue,  thick]plot(\x, {5*(\x) - 1}); %(5,1).
% Obstacle vertices.
\draw[color=olive, thick] plot(\x, {3*(\x)});
\draw[color=teal, thick] plot(\x, {4*(\x) + 3}); 
\draw[color=violet, thick] plot(\x, {\x+3}); 
\end{scope}

\end{scope}

%%3rd map (overlap)
\begin{scope}[shift={(24,0)}, scale=.65,domain=-7.2:7.2]

% Grid.
\draw[very thin,color=gray] (-7,-7) grid (7,7); 
\draw[dashed] (-7.2,0) -- (7.2,0) node[right] {}; 
\draw[dashed] (0,-7.2) -- (0,7.2) node[above] {}; 

%% Dualization of P.
% Fill in forbidden faces.
\begin{scope}
\clip (-7.2,-7.2) rectangle (7.2,7.2);
\pgfsetfillopacity{0.5}
% Overlap of more than 4 double wedges in P.
\draw[fill=blue, ultra thin] (-3,0) -- (-1/3, 8/3) -- (2/5, 6/5) -- (0,0) -- (-3,0) -- cycle;
% Overlap of more than 2 double wedges in B.
\draw[fill=blue, ultra thin] (-7,4) -- (-1,4) -- (1,0) -- (-7, 0) -- (-7,4) -- cycle;
\end{scope}

% Boundary and extended vertices
\clip (-7.2,-7.2) rectangle (7.2,7.2);
\draw[color=red, thick] plot (\x, {-2*(\x) -1}); 
\draw[color=pink, thick] plot (\x, {-2*(\x) +2}); 
\draw[color=yellow, thick] plot (\x, {(\x) - 1});
\draw[color=lime, thick] (-7.2, -1) -- (7.2, -1) ; %(0,1);
\draw[color=green, thick] (-7.2,0) -- (7.2, 0) ;
\draw[color=orange, thick] (-7.2, 4) -- (7.2, 4);
\draw[color=cyan, thick] plot (\x, {5*(\x)+4});
\draw[color=blue,  thick]plot(\x, {5*(\x) - 1});
% Obstacle vertices.
\draw[color=olive, thick] plot(\x, {3*(\x)});
\draw[color=teal, thick] plot(\x, {4*(\x) + 3}); 
\draw[color=violet, thick]plot(\x, {\x+3}); 

%% Valid sweep path.

\node at (0.5,-6.2)  {start face} ;
\node at  (-0.65,6.2) {end face}; 
\draw[|->, very thick, densely dotted] (.55,-5.75) -- (1.75,-.75) -- (1.75,.25) -- (1.75,2.75) -- (-.35,  5.75);

\end{scope}

\end{tikzpicture}
\end{center}

\caption{\small{Determining if an environment is strictly two-sweepable. 
Extend all reflexive vertices to obtain $P'$ (left) and dualize it (center), taking forbidden faces to be those in which more than four double wedges overlap.
Overlay the forbidden faces from the dual of $B$. The valid path shown establishes that $P$ is strictly two-sweepable.}}
\label{fig:combineddual}
\end{figure}

% Motivation for the algorithm.
We conclude by describing how to extend Algorithm~\ref{alg:ssweepfind}  to an algorithm for identifying
strictly two-sweepable environments.
Let $P$ be a polygonal environment containing obstacles and let $B = \partial P$ be its boundary polygon.
Let $P\rq{}$ and $B\rq{}$ be the extensions of $P$ and $B$, respectively, obtained by adding vertices
by extending edges at all reflex vertices of $P$ and $B$.
The environment $P$ is strictly two-sweepable if (1) a line can be swept across $P$
such that the cross sections of $P$ with respect to the line consist of at most two disjoint convex sets,
(2) the cross sections of $B$ are convex, and (3) no point of $P$ is swept more than once.
In order to decide whether $P$ is strictly two-sweepable, we must
determine (a) if the dual arrangement
of $P\rq{}$ admits a sweep-path that avoids all intersections of more than four double wedges, and (b)
if the dual arrangement of $B\rq{}$ admits a sweep-path avoiding all intersections of more than two double wedges. This is summarized in Algorithm \ref{alg:m2sfind}.

\begin{algorithm}
\caption{Strictly Two-Sweepable Path Search}
\label{alg:m2sfind}
\begin{algorithmic}
\item[] Given an environment $P$ with boundary polygon $B=\partial P$.
\item[] Extend all reflex vertices to obtain $P\rq{}$ (with $m$ vertices) and $B\rq$.
\item[] Compute dual arrangements $D_{P\rq{}}$ and $D_{B\rq}$.
\item[] Identify all forbidden faces of $D_{P\rq{}}$.
\item[] Identify all forbidden faces of $D_{B\rq}$ and overlay on $D_{P\rq}$.
\item[] Apply depth first search on the faces of $D_{P\rq{}}$ to find a path avoiding
all forbidden faces of $D_{P\rq}$ and $D_{B\rq}$ crossing the fewest lines.
\item[] If this path crosses exactly $m$ lines then $P$ is strictly two-sweepable.
\end{algorithmic}
\end{algorithm}

\section{Conclusion}
\label{sec:conc}

We have characterized when one pursuer can capture an evader is an environment with a single obstacle. An immediate question that remains to be answered is:  Under what conditions can one pursuer win in an environment with multiple obstacles? For example, our proof no longer holds in the case of two obstacles $H_1, H_2$, even when   $\hull(H_1 \cup H_2) \leq 2$. Indeed, we could have a long zig-zagging alleyway between the obstacles. This would allow the pursuer to sit in the alley, and force the pursuer to give up his guarding position of the convex hull. To forbid such a pathological environment, it would be reasonable to enforce a  \emph{minimum feature size}  (cf.~\cite{klein+suri}), meaning that no two vertices are within unit distance of one another. This simplifying assumption should make the two-obstacle case tractable.

The main open question regarding the lion and man game in polygonal environments is  to fully characterize environments that are two-pursuer-win.
In this work, our focus has been to give a characterization for environments in which a leapfrogging strategy is effective.  Theorem \ref{thm:leapfrog} gives a very general description of the required  family of nested subregions. The dual polygon algorithm in~Section \ref{sec:sweep} identifies one such family, namely strictly two-sweepable environments. It would be interesting to develop an algorithm that can detect when an environment has a leapfrog decomposition, or at least construct other types of leapfrog decompositions. 

\section{Acknowledgments}

This work was supported in part by the Institute for Mathematics and its Applications and in part by NSF Grant DMS-1156701. Volkan Isler was supported in part by NSF Grant IIS-0917676.

%%%%%%%% Bibliography %%%%%%%%%%%%
%%%%%%%%%%%
\bibliography{pursuit}

\end{document}